\documentclass[11pt,a4paper]{article}
\usepackage{fullpage}

\usepackage[utf8]{inputenc}
\usepackage{amsthm,amssymb}
\usepackage[leqno]{amsmath}
\usepackage{authblk}
\usepackage{comment}
\usepackage{algorithm}
\usepackage{algpseudocode}
\usepackage{tcolorbox}
\usepackage{mathtools,thmtools}
\usepackage[mathcal,mathscr]{eucal}
\usepackage{epsfig,graphicx,graphics,color}
\usepackage{caption}
\usepackage{subcaption}
\usepackage{enumerate}
\usepackage{enumitem, crossreftools}
\usepackage[sort]{cite}
\usepackage{titling}
\usepackage{algorithm}
\usepackage{tikz}
\usepackage{hyperref}

\hypersetup{
    colorlinks=true,       
    linkcolor=blue,        
    citecolor=red,         
    filecolor=magenta,     
    urlcolor=cyan,         
    linktocpage=true
}

\newtcolorbox{probbox}{arc=6pt,
                      colback=white!100,
                      colframe=black!50,
                      before skip=6pt,
                      after skip=6pt,
                      boxsep=1pt,
                      left=6pt,
                      right=6pt,
                      top=4pt,
                      bottom=4pt}

\newcommand{\decprob}[3]{
   \begin{center}%
    \begin{minipage}{0.96\linewidth}%
      \begin{probbox}
      \textsc{#1}\\[0.2ex]
      \textbf{Input:} #2\\[0.2ex]
      \textbf{Question:} #3
      \end{probbox}
    \end{minipage}%
  \end{center}
}

\newcommand{\searchprob}[3]{
   \begin{center}%
    \begin{minipage}{0.96\linewidth}%
      \begin{probbox}
      \textsc{#1}\\[0.2ex]
      \textbf{Input:} #2\\[0.2ex]
      \textbf{Goal:} #3
      \end{probbox}
    \end{minipage}%
  \end{center}
}

\theoremstyle{plain}
\newtheorem{Th}{Theorem}

\newtheorem{Claim}[Th]{Claim}
\newtheorem{Cor}[Th]{Corollary}

\theoremstyle{definition}

\newtheorem{Rem}[Th]{Remark}
\newtheorem{?}[Th]{Problem}

\newenvironment{claimproof}[1][\proofname]
  {%
    \proof[#1]%
  }
  {%
    \endproof%
  }

\DeclareMathOperator\cost{cost}
\DeclareMathOperator\OPT{OPT}
\DeclareMathOperator*{\argmin}{arg\,min}

\def\final{0}  
\def\iflong{\iffalse}
\ifnum\final=0  
\newcommand{\knote}[1]{{\color{red}[{\tiny \textbf{Krist\'of:} \bf #1}]\marginpar{\color{red}*}}}
\newcommand{\gnote}[1]{{\color{blue}[{\tiny \textbf{Gergő:} \bf #1}]\marginpar{\color{blue}*}}}
\newcommand{\tnote}[1]{{\color{green}[{\tiny \textbf{Tam\'as:} \bf #1}]\marginpar{\color{green}*}}}
\else 
\newcommand{\knote}[1]{}
\newcommand{\gnote}[1]{}
\newcommand{\tnote}[1]{}
\fi

\title{Manipulating the outcome of\\ stable matching and roommates problems}

\thanksmarkseries{alph}
\author{
Krist\'of B\'erczi\thanks{MTA-ELTE Momentum Matroid Optimization Research Group and MTA-ELTE Egerv\'ary Research Group, Department of Operations Research, E\"otv\"os Lor\'and University, Budapest, Hungary. Email: \texttt{kristof.berczi@ttk.elte.hu}.}
\and
Gergely Cs\'aji\thanks{MTA-ELTE Momentum Matroid Optimization Research Group, Department of Operations Research, E\"otv\"os Lor\'and University, Budapest, Hungary. Email: \texttt{csajigergely@student.elte.hu}.}
\and
Tam\'as Kir\'aly\thanks{MTA-ELTE Momentum Matroid Optimization Research Group and MTA-ELTE Egerv\'ary Research Group, Department of Operations Research, E\"otv\"os Lor\'and University, Budapest, Hungary. Email: \texttt{tamas.kiraly@ttk.elte.hu}.}
}

\begin{document}
\maketitle

\begin{abstract}
The stable marriage and stable roommates problems have been extensively studied due to their high applicability in various real-world scenarios. However, it might happen that no stable solution exists, or stable solutions do not meet certain requirements. In such cases, one might be interested in modifying the instance so that the existence of a stable outcome with the desired properties is ensured.

We focus on three different modifications. In stable roommates problems with all capacities being one, we give a simpler proof to show that removing an agent from each odd cycle of a stable partition is optimal. 
We further show that the problem becomes NP-complete if the capacities are greater than one, or the deleted agents must belong to a fixed subset of vertices. 

Motivated by inverse optimization problems, we investigate how to modify the preferences of the agents as little as possible so that a given matching becomes stable. The deviation of the new preferences from the original ones can be measured in various ways; here we concentrate on the $\ell_1$-norm. We show that, assuming the Unique Games Conjecture, the problem does not admit a better than $2$ approximation. By relying on bipartite-submodular functions, we give a polynomial-time algorithm for the bipartite case. We also show that a similar approach leads to a 2-approximation for general graphs.

Last, we consider problems where the preferences of agents are not fully prescribed, and the goal is to decide whether the preference lists can be extended so that a stable matching exists. We settle the complexity of several variants, including cases when some of the edges are required to be included or excluded from the solution. 

\medskip

\end{abstract}

\section{Introduction}
\label{sec:intro}

The stable marriage problem was introduced by Gale and Shapley in 1962 in their seminal paper~\cite{gale1962college}. Since then, an enormous amount of research has been done in the field; see for example~\cite{manlove2013algorithmics} for a survey.

Our work belongs to the line of research that studies modifications of the preference system in order to get a stable solution with certain properties. Preference modifications in stable matching problems have been widely studied, mainly from a strategic point of view, where a given agent or a coalition of agents wants to submit false preferences in order to achieve better results. This setting is also closely related to inverse optimization problems, where we are given a feasible solution to an underlying optimization problem together with a linear weight function, and the goal is to modify the weights as little as possible so that the input solution becomes optimal; for further details on inverse optimization, see e.g. \cite{demange2010introduction,richter2021inverse}.

\paragraph{Previous work.}

Concerning edge modifications, Abraham et al.~\cite{abraham2005almost} considered the problem of finding a matching in the roommates problem that has a minimum number of blocking edges. They showed that the problem is hard and is also hard to approximate. Bir\'o et al.~\cite{biro2010size} studied a similar question for the marriage problem, where the goal is to find a maximum matching with minimum number of blocking edges. Tan~\cite{tan1990maximum} considered the problem of finding a maximum number of disjoint pairs of persons such that these pairs are stable among themselves, and gave an algorithm for determining an optimum solution.

Roth~\cite{roth1982economics} showed that there is no matching mechanism that produces a stable matching and is strategy-proof for every agent, but the Gale-Shapley algorithm is at least strategy-proof for one side. Coalitional manipulations were studied by Shen et al.~\cite{shen2018coalition}. Hosseini et al.~\cite{hosseini2022two} investigated cases where the coalition can contain both men and women. Aziz et al.~\cite{aziz2015susceptibility} investigated the case where the agents who want to manipulate the outcome have capacity larger than one, like firms or hospitals. Coalitional manipulations where a fixed subset of edges must be contained in the output of the Gale-Shapley algorithm has been studied by Kobayashi and Matsui~\cite{kobayashi2010cheating} and by Gupta and Roy~\cite{gupta2018stable}.

Modifying the preferences by `bribing' the agents in order to have a stable matching satisfying certain conditions was studied by Boehmer et al.~\cite{boehmer2021bribery}. They examined several variants in terms of manipulative actions and manipulation goals in bipartite preference systems. Later, Eiben et al.~\cite{eiben2021preference} extended their work to the capacitated case. Finding a matching that is robust with respect to small modifications of the preferences have been studied by Chen et al.~\cite{chen2021matchings}.

Stable matchings with forced and forbidden edges were considered by several papers, see Dias et al.~\cite{dias2003stable}, Cseh and Manlove~\cite{cseh2016stable}, Cseh and Heeger~\cite{cseh2020stable}. 

\paragraph{Our results.}

One of the difficulties in markets that can be described as stable roommates problems is that they might not have a stable solution at all. To overcome this difficulty, a natural idea is to exclude certain agents or contracts from the market in order to guarantee the existence of a stable outcome. Both cases have been studied already, we concentrate on the deletion of agents. The latter problem was previously solved by Tan~\cite{tan1990maximum} who gave an algorithm based on Irving's algorithm~\cite{irving1985efficient}. We give a significantly simpler approach to show that deleting an agent from each odd cycle in a stable partition is optimal. We also study the case when the cost of excluding someone might differ from agent to agent, or when only agents of some fixed subset are allowed to be removed, and show that these problems are NP-hard as well as when agents have nonnegative integer capacities.

In certain applications, it might happen that the central authority already has a most preferred outcome and it wants to make this outcome acceptable to the agents, too, by possibly compensating them. We concentrate on the problem where to goal is to bribe the agents in order to make a given outcome stable. We consider a more general framework than the one in~\cite{boehmer2021bribery} and~\cite{eiben2021preference}. In our framework, the preferences are given with real numbers, and instead of swap distance we work with the $\ell_1$-norm. This model allows to express that the differences between two adjacent agents in a preference list might differ a lot, whereas in swap distance interchanging any two of them would count as only one swap. In addition, this approach also makes possible to have ties among agents. We show that the problem is solvable in bipartite graphs in the capacitated setting, even when upper and lower bounds are given on the changes and the compensation cost differs for each agent. We also consider the roommate case which has not yet been investigated before. We show that the problem is NP-hard, but it admits a 2-approximation.

Finally, we investigate problems where the preferences are manipulated by the agents themselves in order to force or avoid certain contracts to form. While previous papers only considered manipulations with respect to the outcome of the Gale-Shapley algorithm, we investigate more general cases. Our goal is to ensure that there exists a stable matching with respect to the modified preferences.
This can be practically important, for example when a stable matching is chosen randomly, not by the Gale-Shapley algorithm. We also study cases where there are forbidden edges instead of forced ones, or some places in the preferences are fixed and cannot be manipulated. We show that almost all problems arising this way are NP-complete even in rather restricted settings. 
On the positive side, we give an algorithm for the case when certain places in the preference lists are not allowed to be manipulated by the agents. 
\medskip

The rest of the paper is organized as follows. Basic notation and definitions are introduced in Section~\ref{sec:prelim}. Section~\ref{sec:del} focuses on problems where the goal is to ensure the existence of a stable matching through the deletion of a subset of vertices. Preference modifications in order to make a given outcome stable are discussed in Section~\ref{sec:opt}. Finally, in Section~\ref{sec:ext}, we investigate preference extension problems.

\section{Preliminaries}
\label{sec:prelim}

\paragraph{Basic notations.}

We denote the sets of \emph{real}, \emph{nonnegative real}, \emph{integer}, and \emph{nonnegative integer} numbers by $\mathbb{R}$, $\mathbb{R}_+$, $\mathbb{Z}$, and $\mathbb{Z}_+$, respectively. For a positive integer $k$, we use $[k]\coloneqq \{1,\dots,k\}$. 

Given a ground set $S$ and subsets $X,Y\subseteq S$, the \emph{difference} of $X$ and $Y$ is denoted by $X\setminus Y$. If $Y$ consists of a single element $y$, then $X\setminus\{y\}$ and $X\cup\{y\}$ are 
abbreviated by $X-y$ and $X+y$, respectively. We call a $3$-element subset a \emph{$3$-set} for short. Some hardness proofs in Section~\ref{sec:ext} will be based on the following NP-complete problems, see~\cite{garey1979computers}. 

\decprob{x3c}{
A family $\mathcal{C}=\{C_1,\dots ,C_{3n}\}$ of 3-sets of a ground set $a_1,\dots ,a_{3n}$.
}{
Does there exist $n$ 3-sets in $\mathcal{C}$ that partition $\{ a_1,\dots ,a_{3n}\}$?
}

\decprob{3sat}{
A Boolean formula $\Phi$ in conjunctive normal form where each clause is limited to at most three literals. 
}{
Does there exist a truth assignment to the variables for which $\Phi$ is true? 
}

Let $G=(V,E)$ be an undirected graph, $S\subseteq V$ be a subset of vertices, and $F\subseteq E$ be a subset of edges. The graphs obtained by \emph{deleting the vertices in $S$} (together with the edges incident to them) or \emph{the edges in $F$} are denoted by $G-S$ and $G-F$, respectively. We call $S$ a \emph{vertex cover} if $G-S$ has no edges. For a vertex $u\in V$, the \emph{set of neighbors of $u$} is denoted by $N(u)=\{v\in V\mid uv\in E\}$. The \emph{set of edges in $F$ incident to a vertex} $v\in V$ is denoted by $F(v)$. The following is a well-known NP-hard problem.

\searchprob{vertex-cover}{
A graph $G=(V,E)$. 
}{
Find a vertex cover of minimum size. 
}

\paragraph{Submodular and bipartite-submodular set functions.} 

A set function $f\colon 2^V\to R$ is called \emph{submodular} if $f(X)+f(Y)\ge f(X\cup Y)+f(X\cap Y)$ holds for every $X,Y\subseteq V$. It is known that this is equivalent to $f(X+a)+f(X+b)\geq f(X+a+b)+f(X)$ for every $X\subseteq V$ and distinct $a,b\in V\setminus X$. A function $f$ defined on the vertex set of a bipartite graph $G =(A,B; E)$ is \emph{bipartite-submodular} if there exist two submodular functions $f_A$ and $f_B$ defined on $A$ and $B$, respectively, such that $f(X)=f_A(X\cap A)+f_B(X\cap B)$ for every $X\subseteq A\cup B$. The following result appeared in~\cite{hochbaum2018complexity}.

\begin{Th}[Hochbaum]
\label{thm:submod}
Given a bipartite graph $G=(A,B;E)$ and a polynomial-time computable submodular cost function $c\colon 2^{A\cup B}\to E$, the minimum cost vertex cover problem can be 2-approximated in polynomial time. Furthermore, if $c$ is also bipartite-submodular, then a minimum cost vertex cover can be found in polynomial time.
\end{Th}

\paragraph{Stable matchings.}

Let $G=(V,E)$ be an undirected simple graph, and let $<_u$ be a linear order on $N(u)$ for each $u \in V$. The pair $(G,>)$ is called a \emph{stable roommates instance}. If $G$ is bipartite, then $(G,>)$ is a \emph{stable marriage instance}. If ties are allowed, i.e., $<_u$ are weak orders, then $(G,>)$ is a \emph{stable roommates instance with ties}. For a stable roommates instance $(G=(V,E),>)$ and a matching $M\subseteq E$, an edge $uv \in E\setminus M$ is \emph{blocking} if $u$ prefers $v$ to its partner in $M$ (or has no partner in $M$) and $v$ prefers $u$ to its partner in $M$ (or has no partner in $M$). If no blocking edge exists, then $M$ is a \emph{stable matching}.

Let $(G,>)$ be a stable roommates instance. A \emph{stable partition of $(G,>)$} is a permutation $\pi\colon V\to V$ such that for each $u\in V$, (i) if $\pi(u)\ne \pi^{-1}(u)$, then $u\pi (u),u\pi^{-1}(u)\in E$ and $\pi (u)>_u\pi^{-1}(u)$, and (ii) for each $v$ adjacent to $u$, if $\pi (u)=u$ or $v>_u\pi^{-1}(u)$, then $\pi^{-1}(v)>_v u$. The cycles of the permutation $\pi$ partition the agents into \emph{odd} and \emph{even cycles} depending on the parity of their lengths.  We call an agent $u$ a \emph{singleton in $\pi$} if $\pi (u)=u$. The following fundamental result is due to Tan~\cite{tan1991necessary}. 

\begin{Th}[Tan]\label{thm:tan}
Any stable roomates instance $(G,>)$ admits a stable partition $\pi$, and such a partition can be found in $\mathcal{O}(|V|^2)$ time. Furthermore, each stable partition has the same set of singleton agents and the same set of odd cycles. There exists a stable matching if and only if there exist a stable partition without odd cycles of length $\ge 3$.
\end{Th}

We will also need the notion of weak stability. Let $(G=(V,E),>)$ be a stable roommates instance with ties, and let $M\subseteq E$ be a matching. An edge $uv \in E\setminus M$ \emph{strongly blocks $M$} if $u$ strictly prefers $v$ to its partner in $M$ (or has no partner in $M$) and $v$ strictly prefers $u$ to its partner in $M$ (or has no partner in $M$). If no strongly blocking edge exists, then $M$ is a \emph{weakly stable matching}.

We also consider cases where the vertices of $G$ have nonnegative integer capacities $q(v)$. In such an instance, we call $M\subseteq E$ a \textit{$q$-matching} if $|M(v)|\le q(v)$ for each $v\in V$.
A vertex is \emph{unsaturated} in $M$ if $|M(v)|<q(v)$.
We say that an edge $uv\in E\setminus M$ \textit{blocks} $M$ if both $u$ and $v$ are either unsaturated by $M$ or strictly prefer the other to some of their partners in $M$. The $q$-matching $M$ is called a \textit{stable $q$-matching} if it is not blocked by any edge. We also say that an edge $uv\in E\setminus M$ is \textit{dominated} at $u$ if $|M(u)|=q(u)$ and for all $xu\in M(u)$, we have $x >_u v$. Clearly a $q$-matching $M$ is stable if and only if each $uv\in E\setminus M$ is dominated at either $u$ or $v$.

The hardness of the following problem, proved by Manlove et al.~\cite{Manlove_etal2002}, is a useful tool for proving NP-completeness of various stable matching problems.

\decprob{com-smti}{An instance $(G=(A,B;E),>)$ of the stable marriage problem in which ties only occur at some vertices in $B$ with exactly two neighbors.}{ Is there a perfect weakly stable matching?}


The following two versions of the stable marriage with ties problem was proven to be NP-complete by Cseh and Heeger \cite{cseh2020stable}.

\decprob{1-forced-smti}{
An instance $(G=(A,B;E),>)$ of the stable marriage problem with ties and a forced edge $e\in E$.
}
{
Is there a weakly stable matching containing $e$?
}

\decprob{1-forbidden-smti}{
An instance $(G=(A,B;E),>)$ of the stable marriage problem with ties and a forbidden edge $f\in E$.
}
{
Is there a weakly stable matching not containing $f$?
}


A \emph{one-sided preference model} involves a bipatite graph $G=(A,B;E)$ and linear orders $>_a$ on $N(a)$ for every $a \in A$. A matching $M$ is \emph{Pareto-optimal} if there is no matching $M'$ that is better than $M$ for some agent in $A$, and at least as good as $M$ for the other agents in $A$. Some hardness proofs in Section~\ref{sec:ext} will be based on the following NP-complete problem, introduced by Abraham et al.~\cite{abraham2004pareto}. 

\decprob{min-pom}{
A bipartite graph $G=(A,B;E)$ with preferences for each $a\in A$ and a non-negative integer $k$.
}{
Does there exist a Pareto-optimal matching of size at most $k$?
}

Given a bipartite graph $G=(A,B;E)$ with preferences for each $a\in A$ and a matching $M$, a \emph{coalition} is a cycle $(a_1,b_1,a_2,\dots ,b_r)$, such that $a_ib_i\in M$ for each $i\in [r]$ and each agent prefers $b_{i+1} $ to $b_i$. The matching is \emph{non-wasteful} if there is no element of $B$ that is not covered by $M$, but is preferred by an agent in $A$ to his current partner. Abraham et al. verified the following.

\begin{Th}[Abraham et al.] \label{thm:pareto}
A matching $M$ (in the one sided preference model) is Pareto-optimal if and only if it is maximal, non-wasteful, and coalition-free.
\end{Th}

\section{Vertex deletion}
\label{sec:del}

In this section, we study the problem of removing vertices so that the remaining graph has a stable matching. Two variants are considered: in the first one, the aim is to remove the minimum number of vertices, while in the second one, an arbitrary number of vertices can be removed, but only from a given vertex subset $T \subseteq V$.  

Let $G=(V,E)$ be a graph with strict preferences $>_v$ for each vertex $v \in V$. A vertex set $S\subseteq V$ is \emph{removable} if $G-S$ admits a stable matching. When vertex capacities $q(v)$ are considered, a vertex set $S\subseteq V$ is \textit{removable} if $G-S$ admits a stable $q$-matching.

\decprob{vertex-del-sr}{
A graph $G=(V,E)$ with strict preferences $>_v$ on the vertices and an integer $k\in\mathbb{Z}_+$.
}{
Does there exist a removable set $S\subseteq V$ such that $|S|\leq k$?
}

\decprob{vertex-subset-del-sr}{
A graph $G=(V,E)$ with strict preferences $>_v$ on the vertices and a vertex subset $T\subseteq V$.
}{
Does there exist a removable set $S\subseteq T$?
}

\decprob{vertex-del-sqm}{
A graph $G=(V,E)$ with strict preferences $>_v$ and capacities $q(v)$ on the vertices, and an integer $k\in\mathbb{Z}_+$.
}{
Does there exist a removable set $S\subseteq V$ such that $|S|\leq k$?
}

The analogous edge-deletion problems are also of interest, and have already been studied in the literature. The minimization version is equivalent to finding a matching that has the smallest possible number of blocking pairs. This problem was shown to be NP-hard by Abraham et al.~\cite{abraham2005almost}, even for complete preference lists. The subset version is equivalent to the following: given an edge subset $F\subseteq E$, is there a matching in $G$ for which every blocking edge is in $F$? This problem is called \emph{stable roommates with free edges}, and was shown to be NP-complete by Cechl{\'a}rov{\'a} and Fleiner~\cite{egres-09-01}.

Among the vertex deletion problems, \textsc{vertex-del-sr} was considered by Tan~\cite{tan1990maximum} who gave an algorithm based on Irving's algorithm \cite{irving1985efficient}. In the following, we give a significantly simpler algorithm for \textsc{vertex-del-sr}. We also show that \textsc{vertex-subset-del-sr} is NP-complete. The latter implies that finding a minimum weight removable set is NP-complete even for vertex weights in $\{0,1\}$, and it is inapproximable, since it is hard to decide if there is a removable set of weight 0.



We show that it is very easy to determine a minimum cardinality removable vertex set using a stable partition. The algorithm is presented as Algorithm~\ref{alg:del-vert}; it chooses an arbitrary vertex from each odd cycle in the stable partition. Note, however, that a removable set does not necessarily contain a vertex from each odd cycle, so optimality is not obvious. We show the correctness of the algorithm in the proof of the theorem below.

\begin{algorithm}[h!]
  \caption{Algorithm for \textsc{vertex-del-sr}.} \label{alg:del-vert}
  \begin{algorithmic}[1]
    \Statex \textbf{Input:} A graph $G=(V,E)$ with strict preferences $>_v$ on the vertices.
    \Statex \textbf{Output:} A smallest removable set $S\subseteq V$ and a stable matching $M$ in $G-S$. 
    \State Set $S\leftarrow \emptyset$ and $M\leftarrow\emptyset$.
    \State Find a stable partition $\pi$ using Theorem~\ref{thm:tan}.  
    \State Let $C_1,\dots ,C_k$ be the set of odd cycles of $\pi$ of length $\ge 3$, and let $C_{k+1},\dots ,C_{l}$ be the set of even cycles of $\pi$. Let $\{ v^i_0,\dots,v^i_{j_i}\}$ denote the vertices of $C_i$, where $\pi(v^i_{j_i})=v^i_0$, and $\pi(v^i_j)=v^i_{j+1}$ otherwise.
    \For{$i=1,\dots,k$}
        \State Set $S\leftarrow S+v^i_0$ and $M\leftarrow M\cup\{v_{2j-1}v_{2j}\mid 1\leq j\leq j_i/2\}$.
    \EndFor
    \For{$i=k+1,\dots ,l$}
        \State Set $M\leftarrow M\cup\{v_{2j-2}v_{2j-1}\mid 1\leq j\leq j_i/2\}$.
    \EndFor
    \State \textbf{return} $S$, $M$ \label{st:end}
  \end{algorithmic}
\end{algorithm}

\begin{Th}
Algorithm~\ref{alg:del-vert} runs in in polynomial time and finds a smallest removable set.
\end{Th}
\begin{proof}
By Theorem~\ref{thm:tan}, the algorithm has polynomial running time. We prove the correctness in two steps. First, we show that Algorithm~\ref{alg:del-vert} outputs a removable set.

\begin{Claim}\label{cl:feas}
Algorithm~\ref{alg:del-vert} outputs a set $S\subseteq V$ and a matching $M$ that is stable in $G-S$.
\end{Claim}
\begin{claimproof}
Let $uv$ be an edge of $G-S$ not in $M$; we show that $uv$ is not a blocking edge. Recall that $\pi$ denotes the stable partition of $(G,>)$ obtained by applying Tan's algorithm. 

Assume first that at least one of $u$ and $v$ is a singleton in $\pi$, say, $\pi(u)=u$. Then for any of its neighbors $w\in N(u)$, we have $\pi(w)\neq w$, $\pi(w)>_w u$, and $\pi^{-1}(w)>_w u$ since $\pi$ is a stable partition. In particular, this holds for $v$. Since $v$ is matched to either $\pi(v)$ or $\pi^{-1}(v)$ in $M$, $uv$ is not a blocking edge. 

Consider the case when $\pi(u)\neq u$ and $\pi(v)\neq v$. Since every vertex in $V\setminus S$ that is not a singleton in $\pi$ gets a partner in $M$, both $u$ and $v$ are matched by the algorithm. If $u\neq\pi(v)$ and $v\neq\pi(u)$, then $\pi(v)>_v u,\pi^{-1}(v)>_v u$ or $\pi(u)>_u v,\pi^{-1}(u)>_u v$ hold as $\pi$ is a stable partition. As $u$ is matched to one of the vertices $\pi(u),\pi^{-1}(u)$ and $v$ is matched to one of the vertices $\pi(v),\pi^{-1}(v)$, $uv$ is not a blocking edge. Therefore we may assume that $\pi(u)=v$. However, as $u$ and $v$ are not in $S$ and $uv$ is not in $M$, $v$ is matched to $\pi(v)$ in $M$. By $\pi(v)>_v u$, $uv$ is not a blocking edge. 
\end{claimproof}

It remains to show that the size of a removable set cannot be smaller than the number of odd cycles in a stable partition. 

\begin{Claim}\label{cl:opt}
The minimum size of a removable set is equal to the number of odd cycles of length at least $3$ in any stable partition $\pi$.
\end{Claim}
\begin{claimproof}
Let $S^*$ be a removable set, and let $M^*$ be a stable matching in $G-S^*$. Since $\pi$ is a stable partition, every odd cycle of length at least $3$ in $\pi$ must either have a vertex in $S^*$ or a vertex $u$ matched in $M^*$ to a partner better than $\pi^{-1}(u)$. Indeed, if this does not hold for a cycle, then there is a vertex $v$ in the cycle that is unmatched or matched to a worse partner than $v\pi (v)$ and $v\pi^{-1}(v)$. But then $\pi^{-1}(v)v$ blocks, since $\pi^{-1}(v)$ is not matched to anyone that is better than $v$ by our assumption. We call a vertex $v \in V\setminus S^*$ \emph{out-dominated} if it is matched in $M^*$ to a vertex $u\ne \pi (v)$ with $u>_v\pi^{-1}(v)$. A vertex $u$ is an \emph{in-dominator} if it is the partner in $M^*$ of an out-dominated vertex $v$. Note that $u$ prefers $\pi(u)$ and $\pi^{-1}(u)$ to $v$ because $\pi$ is a stable partition.

Let $C=\{ v_1,\dots,v_t\}$ be an arbitrary cycle of length at least 2 in $\pi$ such that $\pi (v_i)=v_{i+1}$, and let $v_{i_1},\dots,v_{i_z}$ be the in-dominator vertices in $C$. Observe that these vertices are pairwise non-adjacent on the cycle $C$ as an adjacent pair would block $M^*$. Also, each of $v_{i_1-1},\dots,v_{i_z-1}$ must be either in $S^*$ or out-dominated. Indeed, if for some $1\leq j\leq z$ the vertex $v_{i_j-1}$ is not in $S^*$ and also not out-dominated, then the edge $v_{i_j-1}v_{i_j}$ blocks as $v_{i_j}$ is an in-dominator.

Suppose now that $t$ is odd. We claim that there is at least one vertex in $C\setminus\{v_{i_1-1},\dots,v_{i_z-1}\}$ that is either in $S^*$ or is out-dominated. To see this, consider the paths obtained by removing the vertices $\{v_{i_1-1},v_{i_1},\dots,v_{i_z-1},v_{i_z}\}$ from $C$. Since $t$ is odd, one of the paths contains an odd number of vertices; we may assume that this path is $v_1,\dots ,v_p$ for some odd $p$. If none of these vertices is deleted or out-dominated, then at least one them, say $v_i$, must be unmatched or matched in $M^*$ to a vertex worse than $\pi^{-1}(v_i)$. If $i\ne 1$, then $v_{i-1}v_i$ blocks $M^*$, a contradiction. If $i=1$, then $v_t=\pi^{-1}(v_1)$ is an in-dominator as $v_1$ is the first vertex of the path. This implies that $v_tv_1$ blocks $M^*$, a contradiction again. That is, at least one of the vertices $v_1,\dots ,v_p$ is deleted or out-dominated. 

We conclude that the number of in-dominators in each odd cycle $C$ is at most the sum of the numbers of out-dominated and deleted vertices in $C$ minus one. If $C$ is an even cycle, then the number of in-dominators is at most the sum of the numbers of out-dominated and deleted vertices. Finally, a singleton in $\pi$ cannot be an in-dominator. It follows from the definitions that the number of in-dominators and the number of out-dominated vertices are the same. By combining these observations, we get that $|S^*|$ is at least the number of odd cycles in $\pi$. 
\end{claimproof}

The theorem follows by Claims~\ref{cl:feas} and~\ref{cl:opt}.
\end{proof}

We now show that the subset version of the problem is hard.

\begin{Th}\label{thm:subset-del}
\textsc{vertex-subset-del-sr} is NP-complete.
\end{Th}
\begin{proof}
The problem is clearly in NP. To show hardness, we reduce from \textsc{com-smti}, see the definition in Section~\ref{sec:prelim}.

Consider an instance $(G=(A,B;E), >)$ of \textsc{com-smti}, where $A=\{a_1,\dots,a_n\}$, $B=\{b_1,\dots,b_n\}$, and ties are restricted to a subset $b_1,\dots,b_{\ell}$, each of which have two neighbors. We create an instance of \textsc{vertex-subset-del-sr} as follows, see Figure~\ref{fig:vertex_del2} for an example. For each $a_i\in A$, we add three agents $a'_i$, $x_i$ and $x_i'$. For each $i\in [\ell]$, we add three agents $b_i'$, $b_i''$ and $y_i$. Finally, for $i=\ell+1,\dots,n$, we add an agent $b_i'$. Let $T\coloneqq \{b_i'\mid i \in [n]\}\cup \{b_i''\mid i \in [\ell] \}\cup \{y_i\mid i \in [\ell]\}$.

The preference lists of the vertices $a'_i$ for $i\in [n]$ and $b'_i$ for $i=\ell+1,\dots,n$ are inherited from the original instance. We further extend the preferences as follows:
\begin{itemize} \itemsep0em
    \item we add $x_i>x_i'$ to the end of each $a_i'$'s preference lists,
    \item $x_i$ has preference $x_i'>a_i'$ and $x_i'$ has preference $a_i'>x_i$,
    \item $y_i$ has preference $b_i''>b_i'$,
    \item if $i \in [\ell]$ and $b_i$'s two neighbors were $a_j$ and $a_k$ with $j<k$, then $b_i'$'s preference is $y_i>b_i''>a_j'$ and $b_i''$'s preference is $b_i'>y_i>a_k'$.
\end{itemize}

\begin{figure}[t!]
\centering
\begin{subfigure}[t]{\textwidth}
  \centering
  \includegraphics[width=.25\linewidth]{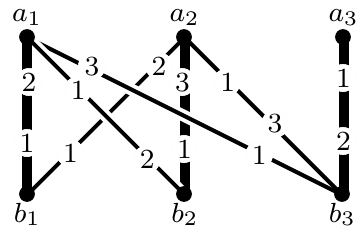}
  \caption{An instance of \textsc{com-smti} with ties at $b_1$. Thick edges form a perfect weakly stable matching $M$.}
  \label{fig:vertex_del1}
\end{subfigure}
\vspace{5pt}

\begin{subfigure}[t]{0.49\textwidth}
  \centering
  \includegraphics[width=.7\linewidth]{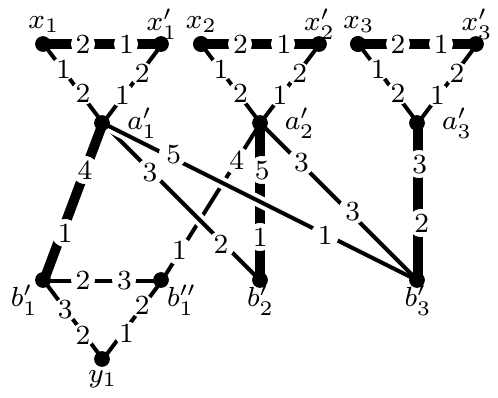}
  \caption{The corresponding \textsc{vertex-subset-del-sr} instance with $T=\{b_1',b_1'',y_1,b_2'\}$. Thick edges form a stable matching $M'$ after deleting $y_1$ and $b_1''$.}
  \label{fig:vertex_del2}
\end{subfigure}\hfill
\begin{subfigure}[t]{0.49\textwidth}
  \centering
  \includegraphics[width=.7\linewidth]{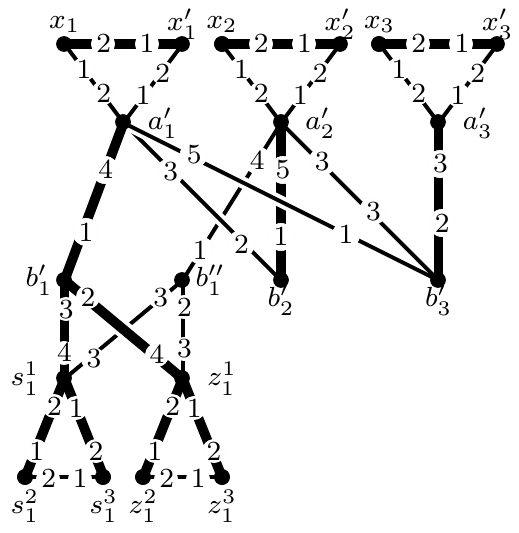}
  \caption{The corresponding \textsc{vertex-del-sqm} instance with capacity $3$ for vertices $b_1'$, $b_1''$, $s^1_1$ and $z^1_1$, and $1$ otherwise. Thick edges form a stable $q$-matching $M'$ after deleting $b_1''$.}
  \label{fig:vertex_del4}
\end{subfigure}
\caption{Illustrations of Theorems~\ref{thm:subset-del} and~\ref{thm:cap-del}. Higher values correspond to higher preferences.}
\label{fig:vertex_del}
\end{figure}

Let $(G',>',T)$ denote the obtained instance of \textsc{vertex-subset-del-sr}. We claim that there exists a perfect weakly stable matching in $(G,>)$ if and only if there exists a removable set $S \subseteq T$ in $(G',>',T)$. First suppose there is a perfect weakly stable matching $M$. We create a matching $M'$ in $G'$ by adding $a_i'b_j'$ or $a_i'b_j''$ for each $a_ib_j\in M$ and $x_ix_i'$ for $i\in[n]$. Then, for every $i \in [\ell]$, we delete $y_i$ and the unmatched copy of $b_i$. Suppose there is a blocking edge to $M'$. It cannot contain $x_i$ or $x_i'$, since each $a_i'$ is with a better partner. On the other hand, if there is a blocking edge of type $a_i'b_j'$, then $j>\ell$, since $b_j'$ is matched with the best remaining partner if $j \leq \ell$. But by the definition of $>'$, this means that $a_ib_j$ is blocking $M$, a contradiction.

For the other direction, suppose there is a set $S \subseteq T$ such that there is a stable matching $M'$ in $G'-S$. Each $a_i$ must be matched to some $b_j'$ or $b_j''$, since otherwise there would be no stable matching because of the cycle $\{ a_i',x_i,x_i'\}$. Also, for any $j \in [\ell]$, at most one of $b_j'$ and $b_j''$ can be matched to some $a_i'$, since otherwise $b_j'b_j''$ would be a blocking edge. Therefore, $M'$ induces a perfect matching $M$ in the original instance. If there is an edge blocking $M$, then the copy of the same edge is blocking $M'$, contradiction. 
\end{proof}

\begin{Rem}
\textsc{vertex-subset-del-sr} is NP-complete even in the special case when $G$ is complete. In the reduction above, we can add the remaining edges to the end of the agents' preference lists, such that they are strictly worse than any original edge. 
\end{Rem}

\begin{Rem}
It also follows from Theorem~\ref{thm:subset-del} that the weighted version of \textsc{vertex-del-sr} is NP-complete too. This can be seen from the fact that if we let the vertices of $T$ have weight 0 and other vertices have weight 1, then there is a removable set of weight 0 if and only if there is a removable set $S\subseteq T$.
\end{Rem}

Finally, we investigate \textsc{vertex-min-del-sqm}. The only difference now is that each vertex $v$ has a nonnegative integer capacity $q(v)$, and the aim is to achieve the existence of a stable $q$-matching by removing a minimum number of vertices. Note that if the aim was to decrease the vertex capacities by a minimum total amount, then the problem would be polynomial-time solvable. Indeed, it could be reduced to \textsc{vertex-del-sr} by making $q(v)$ copies of each vertex. However, if we look for a minimal size removable set $S\subseteq V$, then the problem becomes NP-hard. 

\begin{Th}\label{thm:cap-del}
\textsc{vertex-del-sqm} is NP complete even if the capacities are at most $3$.
\end{Th}
\begin{proof}
The problem is clearly in NP. To show hardness, we reduce from \textsc{com-smti}, see the definition in Section~\ref{sec:prelim}.

Consider an instance $(G=(A,B;E), >)$ of \textsc{com-smti}, where $A=\{a_1,\dots,a_n\}$, $B=\{b_1,\dots,b_n\}$, and ties are restricted to a subset $b_1,\dots,b_{\ell}$, each of which have two neighbors. We create an instance of \textsc{vertex-del-sqm} as follows, see Figure~\ref{fig:vertex_del4} for an example.
For each $a_i\in A$, we add three agents $a_i'$, $x_i$ and $x_i'$ with $q(a_i')=q(x_i)=q(x_i')=1$.
For each $i\in [\ell]$, we add a gadget $G_i$ with eight agents $b_i'$, $b_i''$, $s_i^1$, $s_i^2$, $s_i^3$ and $z_i^1$, $z_i^2$, $z_i^3$. Among these agents, $q(b_i')=q(b_i'')=q(s_i^1)=q(z_i^1)=3$ and $q(s_i^2)=q(s_i^3)=q(z_i^2)=q(z_i^3)=1$. Finally, for $i=\ell+1,\dots,n$, we add an agent $b_i'$ with $q(b_i')=1$. For $i\in [\ell ]$, let $a_{i_1}'$ and $a_{i_2}'$ denote the two neighbors of $b_i$ where $i_1<i_2$.

The preference lists of the vertices $a_i'$ for $i\in [n]$ and $b'_i$ for $i=\ell+1,\dots,n$ are inherited from the original instance. We further extend the preferences as follows:
\begin{itemize} \itemsep0em
    \item we add $x_i>x_i'$ to the end of each $a_i'$'s preference lists,
    \item $x_i$ has preference $x_i'>a_i'$ and $x_i'$ has preference $a_i'>x_i$,
    \item for $i\in [\ell ]$, the agents have the following preferences: \\
    \begin{minipage}{0.3\textwidth}
    \vspace{-0.3cm}
    \begin{align*}
    b_i'\colon{}&{} s_i^1>z_i^1>a_{i_1}' \\
    b_i''\colon{}&{} s_i^1>z_i^1>a_{i_2}' 
    \end{align*}
    \end{minipage}
    \hfill
    \begin{minipage}{0.3\textwidth}
    \vspace{-0.3cm} 
    \begin{align*}
    s_i^1\colon{}&{}  b_i'>b_i''>s_i^2>s_i^3\\
    s_i^2\colon{}&{} s_i^3>s_i^1\\
    s_i^3\colon{}&{} s_i^1>s_i^2
    \end{align*}
    \end{minipage}
    \hfill
    \begin{minipage}{0.3\textwidth}
    \vspace{-0.3cm}
    \begin{align*}
    z_i^1\colon{}&{} b_i'>b_i''>z_i^2>z_i^3\\
    z_i^2\colon{}&{} z_i^3>z_i^1\\
    z_i^3\colon{}&{} z_i^1>z_i^2.    
    \end{align*}
    \end{minipage}
\end{itemize}
Let $(G'=(V',E'),>',q)$ denote the obtained instance of \textsc{vertex-del-sqm} and set $k=\ell$. We claim that there is a perfect weakly stable matching $M$ in $(G,>)$ if and only if there is a removable set $S\subseteq V'$ with $|S|\le k$ in $(G',>',q)$. We start with an important observation.

\begin{Claim}
\label{claim_1}
The gadget $G_i$ has no stable $q$-matching for $i\in [\ell ]$. Furthermore, if $v\in V(G_i)$ is a vertex such that $G_i-v$ admits a stable $q$-matching, then $v\in \{ b_i',b_i''\}$.
\end{Claim}
\begin{claimproof}
Let $G_i$ be one of the gadgets. To see that $G_i$ does not admit a stable $q$-matching, observe that since $q(b_i')=q(b_i'')=q(s_i^1)=q(z_i^1)=3$, all of them can get their two best partners, so the edges $b_i's_i^1$, $b_i'z_i^1$, $b_i''s_i^1$ and $b_i''z_i^1$ must be contained in any stable $q$-matching. But then $s_i^1$ has only one remaining capacity for $s_i^2$ and $s_i^3$. Thus there exists no stable $q$-matching as one of $s_i^3s_i^1$, $s_i^1s_i^2$ and $s_i^2s_i^3$ blocks.

Suppose now that $v\in V(G_i)$ is such that $G_i-v$ admits a stable $q$-matching. If $v\in \{ s_i^1,s_i^2,s_i^3\}$, then, since $z_i^1$ is the best remaining choice for both $b_i'$ and $b_i''$, and $b_i'$ and $b_i''$ are the best two choices for $z_i^1$, the edges $b_i'z_i^1$ and $b_i''z_i^1$ are contained in any stable $q$-matching. But then $z_i^1$ has only one remaining capacity for $z_i^2$ and $z_i^3$, thus there exists no stable $q$-matching. 

The case when $v\in \{ z_i^1,z_i^2,z_i^3\}$ can be proved analogously. 
\end{claimproof}

Suppose there is a perfect weakly stable matching $M$. We create a $q$-matching $M'$ in $G'$ by adding $a_i'b_j'$ or $a_i'b_j''$ for each $a_ib_j\in M$ and $x_ix_i'$ for $i\in[n]$. Then, for every $i \in [\ell]$, we delete the copy of $b_i$ not matched to $A'=\{ a_i'\mid i\in [n]\}$. Then, if $b_i'$ got deleted, we add the edges of $N_i'=\{ b_i''s_i^1,b_i''z_i^1,s_i^1s_i^2,s_i^1s_i^3,z_i^1z_i^2,z_i^1z_i^3\}$ to $M'$ and otherwise we add $N_i''=\{ b_i's_i^1,b_i'z_i^1,s_i^1s_i^2,s_i^1s_i^3,z_i^1z_i^2,z_i^1z_i^3\}$. The number of deleted vertices is then exactly $k$. Suppose there is a blocking edge to $M'$. It cannot contain $x_i$ or $x_i'$, since each $a_i$ is with a better partner. If there is a blocking edge of type $a_i'b_j'$, then $j>\ell$, since $b_j'$ (or $b_j'')$ is matched with all of its partners if $j \leq \ell$. But by the definition of $>'$, this means that $a_ib_j$ is blocking $M$, contradiction. Finally, the edges $s_j^2s_j^3$ and $z_j^2z_j^3$ also do not block $M'$, because $s_j^3$ and $z_j^3$ are matched with their best partner in $M'$ for each $j\le \ell$.

Now assume that there is a set $S \subseteq V'$ with $|S|\le k$ such that there is a stable $q$-matching $M'$ in $G'-S$. Each $a_i'$ must be matched to some $b_j'$ or $b_j''$, since otherwise there would be no stable matching because of the cycle $\{ a_i',x_i,x_i'\}$. Also, for any $j \in [\ell]$, at most one of $b_j'$ and $b_j''$ can be matched to some $a_i'$, since by claim~\ref{claim_1} one of them has to be deleted, otherwise there would be a gadget $G_i$ where at least two vertices are deleted (it is also clear that deleting vertices outside $V(G_i)$ cannot create a stable $q$-matching in $G_i$), contradicting $|S|\le k$. Therefore, $M'$ induces a perfect matching $M$ in the original instance. If there is an edge blocking $M$, then the copy of the same edge is blocking $M'$, contradiction. 
\end{proof}

\section{Preference modifications}
\label{sec:opt}

In this section, we consider problems where the goal is to make a given matching stable by bribing the agents to change their preferences. Each agent $v$ assigns a nonnegative real number $p_v(u)$ to each acceptable partner $u$. By abuse of notation, we extend this notation to edges incident to $v$ by setting $p_v(uv)\coloneqq p_v(u)$, and call this the \emph{value} of the edge $uv$ for $v$. 
The weak preference list of $v$ is derived from these values by declaring $u\ge_vu'$ if and only if $p_v(u)\ge p_v(u')$. The agents might also have positive integer capacities, denoted by $q(v)$ for agent $v$. 

The cost of modifying the preferences can be measured in various ways; here, we concentrate on the case of $\ell_1$-norm. That is, changing the value of $p_v(u)$ from $a$ to $b$ has cost $|a-b|$. If $p'$ denotes the modified preference values, then the \emph{cost of $p'$ (with respect to $p$)} is $\cost(p')\coloneqq \sum_{v\in V}\sum_{u\in N(v)}|p_v(u)-p'_v(u)|$. The goal is to find a new preference matrix $p'$ of minimum cost such that a given b-matching $M$ becomes weakly stable with respect to $p'$.

\searchprob{$\ell_1$-min-pm}
{
A maximal $q$-matching $M$ and a matrix $p$ describing the preferences of the agents.
}
{
Preferences $p'$ of minimum cost such that $M$ is weakly stable with respect to $p'$.
}

First we show that the problem is not only NP-hard in general, but it is also hard to approximate within a factor better than 2.

\begin{Th}\label{thm:hardness}
$\ell_1$-\textsc{min-pm} is NP-hard even if each capacity is one. Furthermore, assuming the Unique Games Conjecture, it cannot be approximated within a factor of $2-\varepsilon$ for any $\varepsilon >0$.
\end{Th}
\begin{proof}
Observe that if each capacity is 1, then there always exists an optimal solution that changes the values only on the edges of $M$. Indeed, lowering the value of a non-matching edge $uv$ to become dominated at vertex $v$ has the same cost as increasing the value of the edge in $M(v)$ for $v$ to dominate $uv$. By doing the latter, we are surely better off as other edges may become dominated too at $v$, while in the first case we do not make any progress on dominating the other edges.

The problem is clearly in NP. To show NP-hardness, we reduce from \textsc{vertex-cover}. Let $G=(V,E)$ be an instance of the minimum vertex cover problem. We construct an instance $(G'=(V',E'),p,M)$ of \textsc{$\ell_1$-min-pm}. Let $V'=\{ v',v''\mid v\in V\}$ and $E'=\{ u'v'\mid  uv\in E\} \cup \{ v'v''\mid v\in V\}$. Furthermore, set $M=\{ v'v''\mid v\in V\}$.  For each vertex $v\in V$, we define $p_{v'}(v'')=p_{v''}(v')=0$ and set all the other $p_{v'}(u')$ values to be 1, see Figure~\ref{fig:vertex_cover}.

\begin{figure}[t!]
\centering
\begin{subfigure}[t]{0.49\textwidth}
  \centering
  \includegraphics[width=.75\linewidth]{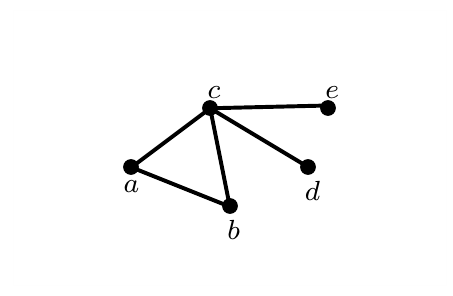}
  \caption{An instance of \textsc{vertex-cover}. Note that $\{a,c\}$ is a vertex cover of minimum size.}
  \label{fig:vertex_cover1}
\end{subfigure}\hfill
\begin{subfigure}[t]{0.49\textwidth}
  \centering
  \includegraphics[width=.75\linewidth]{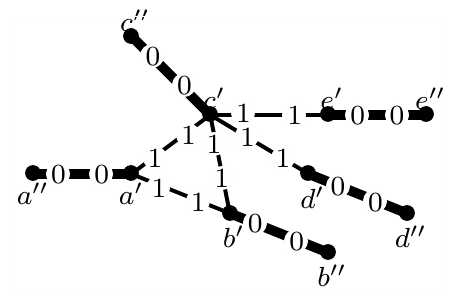}
  \caption{The corresponding \textsc{$\ell_1$-min-pm} instance. If the values of $p_{a'}(a'a'')$ and $p_{c'}(c'c'')$ are increased to $1$, then thick edges form a weakly stable matching.}
  \label{fig:vertex_cover2}
\end{subfigure}
\caption{An illustration of Theorem~\ref{thm:hardness}. Higher values correspond to higher preferences.}
\label{fig:vertex_cover}
\end{figure}

Let $p'$ be a minimum cost solution of the $\ell_1$-\textsc{min-pm} instance thus obtained. By the observation above, we may assume that $p'$ differs from $p$ only on edges of $M$. Since vertices in $\{ v''\mid  v\in V\}$ are matched with their only partner, their preference values do not change. Clearly, for each vertex $v'$, we have either $p'_{v'}(v'')=0$ or $p'_{v'}(v'')=1$. Furthermore, at least one of the end vertices of each original edge increases the value of the incident matching edge to $1$, hence the set of such vertices forms a vertex cover. Vice versa, if we increase the values of the matching edges on the vertices of a vertex cover, then no original edge of $G$ blocks $M$.

By the above, \textsc{minimum vertex cover} in $G$ is equivalent to $\ell_1$-\textsc{min-pm} in $(G',p,M)$, concluding the proof of NP-hardness. Assuming the Unique Games Conjecture, \textsc{vertex-cover} cannot be approximated within a factor better than $2$ by a result of Khoth and Regev~\cite{khot2008vertex}, implying the second half of the theorem.
\end{proof}

As a positive result, we show that $\ell_1$-\textsc{min-pm} is polynomial-time solvable in bipartite graphs. Interestingly, our solution relies on Hochbaum's algorithm for the minimum cost vertex cover problem under bipartite-submodular functions. 

\begin{Th}\label{thm:bipartite}
$\ell_1$-\textsc{min-pm} can be solved in polynomial time in bipartite graphs. 
\end{Th}
\begin{proof}
We reduce $\ell_1$-\textsc{min-pm} to a minimum weight vertex cover problem in a bipartite graph with a polynomial-time computable bipartite-submodular cost function. Recall that, by Theorem~\ref{thm:submod}, the latter problem is solvable in polynomial time. 

Let $G=(A,B;E)$ be the bipartite graph and $M$ be a fixed $q$-matching. We first start with a series of preprocessing steps. If the $q$-matching $M$ is not maximal, then the problem is infeasible and we stop. Otherwise, we iterate through all edges and check whether they are already dominated. If an edge is already dominated, then we delete it. This can be done as there always exists an optimal solution that does not decrease the values on $M$ and does not increase the values on $E\setminus M$. Therefore, we may assume that the $q$-matching $M$ is maximal and all edges outside $M$ are blocking edges.

We construct an an auxiliary bipartite graph $H=(U,W;F)$ as follows. For each vertex $v$, we add $\ell_v$ copies $v_1,\dots,v_{\ell_v}$ of $v$, where $\ell_v$ is the number of edges in $E(v)\setminus M(v)$. These vertices are added to $U$ if $v\in A$ and to $W$ if $v \in B$. For each vertex $v$, consider an arbitrary ordering $e_1^v,\dots,e_{\ell_v}^v$ of the non-matching edges incident to $v$. Then we define the edges of $H$ as follows: $u_iv_j\in F$ if and only if $uv\in E\setminus M$, $uv=e_i^u$, and $uv=e_j^v$.

Now we define an LP for each vertex $v$ saturated by $M$ and each subset $\{ e^v_{j_1},\dots,e^v_{j_k}\}$ of $\{ e_1^v,\dots,e_{\ell_v}^v\}$:
\begin{equation} \label{eq:lp}
\tag{$LP_v(v_{j_1},\dots,v_{j_k})$}
 \begin{array}{rlr}
  \min & \displaystyle \sum_{e\in E(v)}|x(e)-p_v(e)| \\[15pt]
  \text{s.\,t.} & 0\le x(e^v_{j_i})\le \min_{e\in M(v)}  x(e) & \qquad \text{for $1\leq i\leq k$}
 \end{array}
\end{equation}
Though the above formulation is not an LP, both the objective function and the constraints can be linearized using standard techniques. It is not difficult to see that the optimum value of the LP is exactly the minimum cost of dominating all of $\{ e_{j_1},\dots,e_{j_k}\}$ at $v$ by changing the values of the edges. We denote the optimum value by $\OPT_v(v_{j_1},\dots,v_{j_k})$. If $v$ is not saturated by $M$, then $\OPT_v(v_{j_1},\dots,v_{j_k})$ is defined to be $+\infty$, since then no edge can be dominated at $v$. By the maximality of $M$, at least one endpoint of each blocking edge is saturated, so we know that there exists a solution.

For each vertex $v\in A\cup B$, we define a function $f_v$. Let $X\subseteq U\cup W$. If $X\cap \{ v_1,\dots,v_{\ell_v}\} =\emptyset$, then define $f_v(X)\coloneqq 0$. Otherwise, if $X\cap \{ v_1,\dots,v_{\ell_v}\}=\{ v_{j_1},\dots,v_{j_k}\}$, then define $f_v(X)=f_v(\{ v_{j_1},\dots,v_{j_k}\} )\coloneqq OPT_v(v_{j_1},\dots,v_{j_k})$.

\begin{Claim}\label{cl:submod}
For each $v\in A\cup B$, $f_v$ is submodular. 
\end{Claim}
\begin{claimproof}
Our goal is to show that for any set $X\subseteq U\cup W$ and any two distinct elements $x,y\in (U\cup W)\setminus X$, we have $f_v(X+x+y)+f_v(X)\le f_v(X+x)+f_v(X+y)$. If the left hand side is $+\infty$, then $v$ must be unsaturated by $M$ and $X\cup\{x,y\}$ must contain a vertex from $\{ v_1,\dots,v_{\ell_v} \}$, hence the right hand side is also $+\infty$. Thus, we may assume that the left hand side is finite.

If $x\notin \{ v_1,\dots,v_{\ell_v}\}$, then the inequality holds by $f_v(X+x+y)=f_v(X+y)$ and $f_v(X+x)=f_v(X)$, while if $y\notin \{ v_1,\dots,v_{\ell_v}\}$, then the inequality holds by $f_v(X+x+y)=f_v(X+x)$ and $f_v(X+y)=f_v(X)$. 

It only remains to consider the case when $x,y\in \{ v_1,\dots,v_{\ell_v}\}$. Let $e$ and $f$ be the edges of $H$ incident to $x$ and $y$, respectively. By symmetry, we may assume that $p_v(e)\le p_v(f)$. Let $z_1$ be an optimal solution for $LP_v((X+x)\cap \{ v_1,\dots,v_{\ell_v}\} )$ and $z_2$ be an optimal solution for $LP_v((X+y)\cap \{ v_1,\dots,v_{\ell_v}\})$. If one of them is a feasible solution for $LP_v((X+x+y)\cap \{ v_1,\dots,v_{\ell_v}\})$, then since the other is a feasible solution of $LP_v(X\cap \{ v_1,\dots,v_{\ell_v}\})$, we get $f_v(X+x+y)+f_v(X)\le f_v(X+x)+f_v(X+y)$.

Otherwise, $f$ is not dominated when the edge values are set to $z_1$, and $e$ is not dominated when the edge values are set to $z_2$. Since $p_v(e)\le p_v(f)$, this means that $z_2(f)=\min_{g\in M(v)} z_2(g)< z_2(e)=p_v(e)\le p_v(f)$. To see this, observe that $p_v(f)\geq\min_{g\in M(v)} p_v(g)$ as non-blocking edges were deleted during the preprocessing step, and that setting the value of $z_2(f)$ below $\min_{g\in M(v)} z_2(g)$ or setting $z_2(e)\neq p_v(e)$ would be suboptimal. An analogous reasoning shows $z_1(e)=\min_{g\in M(v)} z_1(g)<z_1(f)=p_v(f)$. We distinguish two cases.

If $z_1(e)\le z_2(f)$, then $z_1(e)\le z_2(f)=\min_{g\in M(v)} z_2(g)<z_2(e)=p_v(e)\le z_1(f)$. Consider $z_1'$ that is obtained from $z_1$ by changing $z_1'(e)$ to $p_v(e)$ and $z_2'$ that is obtained from $z_2$ by changing $z_2'(e)$ to $z_2(f)$. Then $z_1'$ is a feasible solution for $LP_v(X\cap \{ v_1,\dots,v_{\ell_v}\})$, while $z_2'$ is a feasible solution for $LP_v((X+x+y)\cap \{ v_1,\dots,v_{\ell_v}\})$. As the sum of the objective values is at most as much as for $z_1$ and $z_2$, we get that $f_v(X+x+y)+f_v(X)\le f_v(X+x)+f_v(X+y)$.

If $z_2(f)\le z_1(e)$, then $z_2(f)\le z_1(e)\le z_2(e)=p_v(e)\le z_1(f)=p_v(f)$. Consider $z_1'$ that is obtained from $z_1$ by changing $z_1'(f)$ to $z_1(e)$ and $z_2'$ that is obtained from $z_2$ by changing $z_2'(f)$ to $p_v(f)$. Then $z_1'$ is a feasible solution for $LP_v((X+x+y)\cap \{ v_1,\dots,v_{\ell_v}\})$, while $z_2'$ is a feasible solution for $LP_v(X\cap \{ v_1,\dots,v_{\ell_v}\})$. As the sum of the objective values is at most as much as for $z_1$ and $z_2$, we get that $f_v(X+x+y)+f_v(X)\le f_v(X+x)+f_v(X+y)$.
\end{claimproof}

Let us define a cost function $c$ on the subsets of $U\cup W$ by $c(X)\coloneqq \sum_{v\in U}f_v(X)+\sum_{v\in W}f_v(X)$. Clearly, since each  function $f_v$ is submodular and the copies of the same vertex are on the same side, $c$ is bipartite-submodular. Note that the value of $f_v(X)$ can be computed in polynomial time by solving the corresponding LP, hence $c$ is also computable in polynomial time. By Theorem~\ref{thm:submod}, we can find a vertex cover of $H$ of minimum cost efficiently. 

By the nonnegativity of $c$ and the fact that each vertex has degree one in $H$, there exists an optimal vertex cover that covers exactly one endpoint of each edge. Such a vertex cover has the same cost as dominating each edge at the chosen endpoint in an optimal way. This implies that the cost of an optimum vertex cover is equal to the cost of an optimal preference modification in the original \textsc{$\ell_1$-min-pm} instance. Furthermore, given an optimal vertex cover, the optimal modified preference values can be determined by solving the corresponding LPs, concluding the proof of the theorem.
\end{proof}

The same reduction can be applied even if the original graph is not bipartite. However, then it is not possible to put each copy of a vertex on the same side of the auxiliary graph $H$ (which is a perfect matching), therefore the cost function $c=\sum f_v$ cannot be written as the sum of two submodular cost functions, which means that, though it is submodular, it is not bipartite-submodular. Therefore, by Theorem~\ref{thm:submod}, we get the following.

\begin{Th}\label{thm:2approx}
There is a polynomial-time 2-approximation algorithm for the \textsc{$\ell_1$-min-pm} problem. \hfill\qed
\end{Th}

Note that, assuming the Unique Games Conjecture, the approximation given by Theorem~\ref{thm:2approx} is best possible by Theorem~\ref{thm:hardness}.

\begin{Rem}
If an arbitrary (not necessarily maximal) $q$-matching $M$ is given and the goal is to modify the preferences in such a way that there exists a stable $q$-matching containing $M$, then the problem becomes difficult. Indeed, deciding the existence of a stable matching containing a given edge in the stable marriage problem with ties is NP-complete, as was proved by Cseh and Heeger~\cite{cseh2020stable}.
\end{Rem}

\begin{Rem}
If the values are not allowed to be changed on the edges of $M$, then the problem becomes very simple even for hypergraphs. Indeed, in such a case one only has to go through the hyperedges $f$ not in $M$ and decrease the value of $p_v(f)$ to $\min_{e\in M(v)} p_v(e)$ at a vertex $v$ where the cost of such a change is minimum, that is, $v\in\argmin_{u \in f} |p_u(f)-\min_{e\in M(u)} p_u(e)|$.
\end{Rem}

The technique described in the proof of Theorem~\ref{thm:bipartite} can be extended to a more general setting. Assume that for each vertex-edge pair $v\in V$ and $e\in E(v)$, nonnegative lower and upper bounds $\ell_v(e),u_v(e)$ are given for the target value of $e$ at $v$. We call this problem $\ell_1$-\textsc{min-pmri} (minimum preference modification with restricted intervals).

\searchprob{$\ell_1$-min-pmri}
{
A maximal $q$-matching $M$, a matrix $p$ describing the preferences of the agents, and nonnegative lower and upper bounds $\ell,u$ on the values of the modified preferences.
}
{
Preferences $p'$ of minimum cost such that $M$ is weakly stable with respect to $p'$ and $\ell_v(e)\leq p'_v(e)\leq u_v(e)$ for each $v\in V$, $e\in E(v)$.
}

Note that this extra condition does not make the problem harder. First, one can apply the following preprocessing step: if $p_v(e)\leq \ell_v(e)$ for some vertex $v$ and $e\in E(v)$, then change the value of $p_v(e)$ to $\ell_v(e)$, and if $p_v(e)\geq u_v(e)$ for some vertex $v$ and $e\in E(v)$, then change the value of $p_v(e)$ to $u_v(e)$. It is not difficult to see that this results in an equivalent problem. After the preprocessing step, the problem can be simply encoded into LPs as before by changing each $LP_v(v_{j_1},\dots,v_{j_k})$ as follows:
\begin{equation*}
 \begin{array}{rll}
  \min & \displaystyle \sum_{e\in E(v)}|x(e)-p_v(e)| \\[15pt]
  \text{s.\,t.} & 0\le x(e^v_{j_i})\le \min_{e\in M(v)}  x(e) & \qquad \text{for $1\leq i\leq k$}\\
     & 0 \le x(e)\le u_v(e) & \qquad \text{for $e\in M$}\\
     & \ell_v(e) \le x(e) & \qquad \text{for $e\in E\setminus M$}
 \end{array}
\end{equation*}

As a consequence, we have the following result. 
\begin{Cor}
 $\ell_1$-\textsc{min-pmri} can be solved in polynomial time in bipartite graphs, and admits a polynomial-time 2-approximation in general graphs.
\end{Cor}

\begin{Rem}
The proposed framework can also be extended to more general cases, for example when the cost for each individual agent might depend not only of the $\ell_1$ norm of the changes, but on the agent itself - e.g., it might be a lot more difficult to get some agents to cooperate than others. This can be modeled by a cost function of the form $\cost(p')=\sum_{v\in V}\sum_{u\in N(v)}\lambda_v|p_v(u)-p'_v(u)|$, where $\lambda_v$ is a constant that depends on the agent $v$. Such a cost function can be encoded in the LP by multiplying the objective of each $LP_v(v_{j_1},\dots,v_{j_k})$ by $\lambda_v$.
\end{Rem}

\section{Preference extension problems}
\label{sec:ext}

\subsection{Fixed and arbitrary preferences}
\label{sec:arb}

Now we turn our attention to problems where the preference list of each agent is either fixed or can be set arbitrarily. Such problems may arise, for example, when dealing with cheating strategies, where some of the agents might submit falsified preferences in order to get a desired stable matching. 

Similar problems have been considered before, see e.g.~\cite{kobayashi2010cheating,gupta2018stable}. Kobayashi et al.~\cite{kobayashi2010cheating} investigated problems where each preference list has to be complete. They showed that if the preferences are fixed for each man, then deciding whether there are preference lists of the women such that a partial matching $M$ is included in the men-optimal stable matching is NP-complete. On the positive side, they showed that if the preferences of all the women unmatched by $M$ are fixed too, then the problem becomes polynomial-time solvable. 

The problems that we consider differ from previous results in several aspects: the completeness of the preference lists is not always assumed, the underlying graph might be non-bipartite, and some of the edges might be forced or forbidden to be contained in a stable matching. It should be emphasized that we are interested not only in men-optimal stable matchings, but in the existence of a stable matching satisfying the constraints in general. 

The problems that are featured in this section are the followings.

\decprob{K-forced-sm-strat}{
A bipartite graph $G=(A,B;E)$, strict preferences $>_a$ for each man $a\in A$, strict preferences $>_b$ for some women $b\in B$, and a partial matching $\mu$ with $|\mu|=K$.
}
{
Are there preferences for the rest of the women such that there exists a stable matching containing $\mu$?
}

\decprob{K-forbidden-sm-strat}{
A bipartite graph $G=(A,B;E)$, strict preferences $>_a$ for each man $a\in A$, strict preferences $>_b$ for some women $b\in B$, and a set of forbidden edges $\tau$ with $|\tau|=K$.
}
{
Are there preferences for the rest of the women such that there exists a stable matching disjoint from $\tau$?
}

\decprob{sr-strat}{
A graph $G=(V,E)$ and strict preferences $>_v$ for some agents $v\in V$.
}{
Are there preferences for the rest of the agents such that there exists a stable matching $M$?
}

First we consider the complexities of \textsc{1-forced-sm-strat} and \textsc{1-forbidden-sm-strat}.

\begin{Th}\label{thm:forced1}
\textsc{1-forced-sm-strat} and \textsc{1-forbidden-sm-strat} are NP-complete even if the preference list is fixed for none of the women.
\end{Th}
\begin{proof}
We start with the proof for \textsc{1-forbidden-sm-strat}; the problem is clearly in NP. To show NP-hardness, we reduce from \textsc{min-pom}, see the definition in Section~\ref{sec:prelim}. 

Given an instance $(G=(A,B;E),>),k$ of \textsc{min-pom}, we construct another bipartite graph $H=(U,W;F)$, where $U$ is the set of men and $W$ is the set of women as follows. Let $n\coloneqq |A|$. We extend the original bipartite graph by adding $n-k$ women $s_1,\dots ,s_{n-k}$ to $B$ that will be called \emph{selectors}. Finally, we add two men $u$ and $u'$ to $A$ and a woman $w$ to $B$. So $U=A\cup \{ u,u'\}$ is the set of men and $W=B\cup \{ w\} \cup \{ s_1,\dots,s_{n-k}\}$ is the set of women in the new instance. The edges are $F=E\cup \{ s_ia_j \mid  i\in [n-k], a_j\in A\} \cup \{ u's_i \mid  i\in [n-k]\}\cup \{ uw,u'w\}$. The preferences are fixed only for the men:
\begin{itemize} \itemsep0em
    \item $u$: $w$,
    \item $u'$: $s_1>\dots >s_{n-k}>w$,
    \item $a_i$: $[N(a_i)]>s_1>\dots >s_{n-k}$,
\end{itemize}
where $[N(a_i)]$ denotes the original neighbors of $a_i$ as in $G$ ranked in their original order, see Figure~\ref{fig:forced_strat}. Finally, let $\tau$ consists of the single forbidden edge $uw$. We claim that there is a Pareto-optimal matching of size at most $k$ in $G$ if and only if the preferences of the women can be set in a way such that there is a stable matching not containing $uv$ in $G'$.

\begin{figure}[t!]
\centering
\begin{subfigure}[t]{0.49\textwidth}
  \centering
  \includegraphics[width=.9\linewidth]{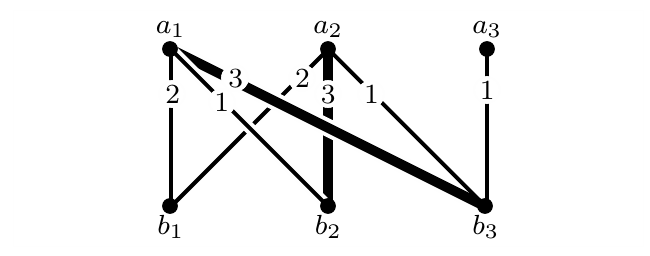}
  \caption{An instance of \textsc{min-pom}. Thick edges form a Pareto-optimal matching $M$.}
  \label{fig:forced_strat1}
\end{subfigure}\hfill
\begin{subfigure}[t]{0.49\textwidth}
  \centering
  \includegraphics[width=.9\linewidth]{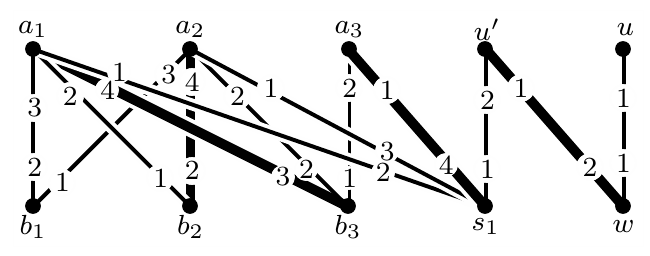}
  \caption{The corresponding \textsc{1-forbidden-sm-strat} instance. Thick edges form a stable matching $M'$ not containing $uw$.}
  \label{fig:forced_strat2}
\end{subfigure}
\caption{An illustration of Theorem~\ref{thm:forced1}. Higher values correspond to higher preferences.}
\label{fig:forced_strat}
\end{figure}

We claim that the original \textsc{min-pom} instance has a solution if and only if the \textsc{1-forbidden-sm-strat} instance thus obtained does. First suppose there exists a Pareto-optimal matching $M$ in $G$ of size $k$. Then extend this matching by assigning $s_1,\dots,s_{n-k}$ to $n-k$ agents in $A$ who are unassigned in $M$ arbitrarily, and assign $u'$ to $w$. Denote this matching by $M'$. Set the preferences of each woman such that they prefer their partner the most if they have any, and arbitrarily otherwise. Clearly, $M'$ does not contain $uw$. Suppose there is a blocking edge $xy$ where $x\in U$ and $y\in W$. Then $y$ can only be an unassigned woman, that is, $y\in B$. This implies that $x\in A$ and either $x$ is unmatched, $x$ is matched with a selector, or $x$ is matched with a $b_i$ in $M'$ that is worse than $y$. In the first two cases $M$ is not maximal, while the third option means that $M$ is not non-wasteful, a contradiction.

For the other direction, suppose there is a preference extension for the women and a matching $M'$ not containing $uw$ that is stable with respect to those preferences. This is equivalent to the case where each woman's preference list is just a tie and there is a weakly stable matching $M'$ not containing $uw$. Since $uw$ does not block, we have $u'w\in M'$. Therefore, each selector woman $s_i$ is assigned, since otherwise they would block with $u'$. Let $M$ be the restriction of $M'$ to $(A,B;E)$. The size of $M$ is at most $k$, because each selector is matched to $A$ and $|A|=n$. We claim that $M$ is non-wasteful and maximal. If it is not maximal, then there is an edge $ab$ that can be added, but $ab$ would block $M'$. Similarly, if an agent $a$ would prefer an unassigned agent $b$ to $M(a)$, then again $ab$ would block $M$.

We conclude that there is a maximal and non-wasteful matching $M$ for the original problem of size at most $k$. If $M$ is coalition-free, then we are done. Otherwise, if there is a coalition $(a_1,b_1,\dots ,a_r,b_r)$, then by changing $M$ to $(M\setminus \{ a_ib_i\mid i\in [r]\})\cup\{ a_ib_{i+1}\mid i\in [r] \}$, we get another matching that covers the same set of vertices, each agent weakly improves, and some agents strictly improve, so it is still maximal and non-wasteful. By iterating this step we get a matching $M$ covering the same set of vertices that is maximal, coalition-free, and non-wasteful, so it is a  Pareto-optimal matching of size at most $k$ by theorem \ref{thm:pareto}.

The forced edge case \textsc{1-forced-sm-strat} can be proved in a similar way, the only difference being that $u$ is dismissed from the construction and we set $\mu$ to contain the single forced edge $u'w$.
\end{proof}

\begin{Rem}
Kobayashi et al.~\cite{kobayashi2010cheating} showed hardness of the problem whether there exist preferences for the women such that the man-optimal stable matching contains a given matching $M$. Our results imply that if the graph is not complete, then the problem remains hard even if the matching consists of a single edge. To see this, observe that in our construction this problem also reduces to deciding whether there are preferences for women such that the man-optimal stable matching matches a given set of women. Indeed, the additional man $u'$ proposes to all of them before proposing to his forced pair $w$, so they all have to reject him, which is only possible if each of them get a partner that can be set to be the best on the preference list of the given woman. Since each stable matching covers the same set of agents, this is equivalent to deciding whether there is a stable matching that matches all the original women.
\end{Rem}

In the proof of Theorem~\ref{thm:forced1}, the reduction was based on a non-complete bipartite graph, e.g., it contained a pendant edge $uw$. We show that the problems remain NP-complete even in complete bipartite graphs. However, in this case the preference lists of some women might be fixed.

\begin{Th}\label{thm:forced2}
\textsc{1-forced-sm-strat} and \textsc{1-forbidden-sm-strat} are NP-complete even if the graph is a complete bipartite graph.
\end{Th}
\begin{proof}
We start with the proof for \textsc{1-forbidden-sm-strat}; the problem is clearly in NP. To show NP-hardness, we reduce from \textsc{com-smti}, see the definition in Section~\ref{sec:prelim}. 

Consider an instance $I=(G=(A,B;E), >)$ of \textsc{com-smti}, where $A=\{a_1,\dots,a_n\}$, $B=\{b_1,\dots,b_n\}$, and ties are restricted to a subset $b_1,\dots,b_{\ell}$, each of which have two neighbors. We construct a \textsc{1-forbidden-sm-strat} instance $I'$ as follows.
We extend the original bipartite graph by adding two men $u$ and $u'$ to $A$ and two women $w$ and $w'$ to $B$. That is, $U=A\cup\{u,u'\}$ and $W=B\cup\{w,w'\}$. Finally, $H$ is a complete graph, that is, $F=\{ xy \mid x\in U, y\in W\}$, see Figure~\ref{fig:forced_complete} for an example. 

Let $W_1=\{ b_i\mid i\in[\ell] \}$ and $W_2=\{ b_i\mid \ell < i\le n\}$. The fixed preferences are the following. Each $a_i$ has the same preferences as in $I$ as the top of its preference list, followed by $w$, and then by the remaining women in any order. Each $b_i\in W_2$ has the same preferences as in $I$ as the top of her preference list, followed by the remaining men in any order. The other fixed preference lists are the following:
\begin{itemize} \itemsep0em
    \item $w$: $u$ is the least preferred, otherwise arbitrary,
    \item $w'$: $u$, followed by $u'$, then the rest in arbitrary order,
    \item $u$: $w$, followed by $w'$, then the rest in arbitrary order,
    \item $u'$: $w'$, then the rest in arbitrary order.
\end{itemize}
Finally, let $\tau$ consists of the single forbidden edge $uw'$. 

\begin{figure}[t!]
\centering
\begin{subfigure}[t]{0.49\textwidth}
  \centering
  \includegraphics[width=.75\linewidth]{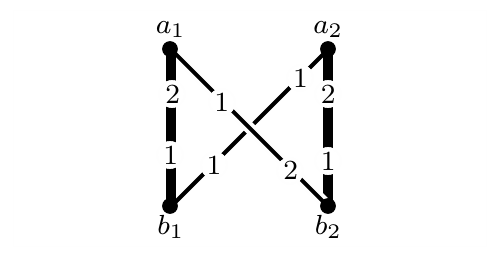}
  \caption{An instance of \textsc{com-smti} with ties at $b_1$. Thick edges form a weakly stable matching $M$.}
  \label{fig:forced_complete1}
\end{subfigure}\hfill
\begin{subfigure}[t]{0.49\textwidth}
  \centering
  \includegraphics[width=.75\linewidth]{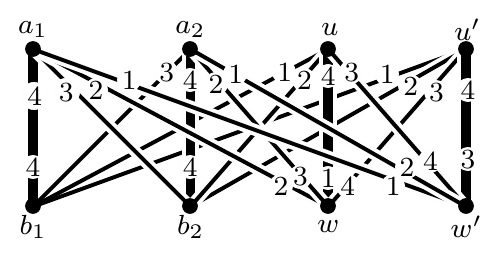}
  \caption{The corresponding \textsc{1-forbidden-sm-strat} instance. Thick edges form a stable matching $M'$, independently from the missing preference values.}
  \label{fig:forced_complete2}
\end{subfigure}
\caption{An illustration of Theorem~\ref{thm:forced2}. Higher values correspond to higher preferences.}
\label{fig:forced_complete}
\end{figure}

We claim that there is a perfect weakly stable matching $M$ in $G$ if and only if the preferences of the remaining women in $H$ can be set in a way such that there is a stable matching $M'$ avoiding $uw'$. First suppose there is a perfect weakly stable matching $M$ in $I$. Let $M'$ be the matching obtained from $M$ by adding $wu$ and $w'u'$. We set the preferences of the women in $W_1$ such that the best man for them is their partner in $M$. We claim that $M'$ is stable. Suppose indirectly that there is a blocking edge $xy$. Then $y$ cannot be $w$ since each man $a_i$ is with an original (and hence better) partner, and $u$ is with a better partner, too. Furthermore, $y$ cannot be $w'$ either, since her only better partner $u$ is with his best choice. Similarly, $y$ cannot be from $W_2$, since they get their most preferred choices. We conclude that $y=b_j \in W_2$ for some $\ell<j\leq n$, so her partner is the same as in $M$. Thus $x=a_i$ for some $1\leq i\leq n$, and $a_i$ must be a man in $A$ who prefers $b_j$ to his partner in $M$. But then $a_ib_j$ would block $M$, a contradiction.

For the other direction, suppose that we can set the preferences of the women in $W_1$ such that there is a stable matching $M'$ not containing $uw'$. Since $uw'$ does not block, we have $uw\in M'$. This implies that $u'w'$ must be in $M'$, too. So the agents corresponding to the original ones are matched among each other. Also, $M'$ has to be a perfect matching, since a man not covered by $M'$ would block with $w$. Let $M$ be the perfect matching obtained by restricting $M'$ to $H-\{u,u',w,w'\}$. If $M$ is not weakly stable in the original instance, then there exists a blocking edge between a woman $b_j$ where $j>\ell$ and a man $a_i$, but then $a_ib_j$ would block $M'$ too, a contradiction.

The forced edge case \textsc{1-forced-sm-strat} can be proved in a similar way, the only difference being that we set $\mu$ to contain the single forced edge $uw$.
\end{proof}

\begin{Rem}
Note that if the graph is a complete bipartite graph and no woman's preference list is fixed, then \textsc{k-forced-sm-strat} and \textsc{k-forbidden-sm-strat} become easily solvable: we just have to find a perfect matching satisfying the edge restrictions if one exists, which is an assignment problem.
\end{Rem}

\begin{Rem}
Theorem~\ref{thm:forced2} also implies NP-hardness of the stable allocation problem with forced or forbidden edges, where only one side has preferences.
\end{Rem}

Our results have an interesting consequence for the stable marriage problem with ties and the incomplete preference lists problem with forced and forbidden edges. 

\begin{Cor}
\textsc{1-forced-smti} and \textsc{1-forbidden-smti} are NP-complete even if each man's preference list is strict and each woman's preference list is just a single tie. The same holds if each man's preference list is strict, each woman's preference list is either strict or a single tie, and all preferences are complete.
\end{Cor}

Now we turn to the complexity of \textsc{sr-strat}.

\begin{Th}
\textsc{sr-strat} is NP-complete even for complete graphs.
\label{sr-strat}
\end{Th}
\begin{proof}
The problem is clearly in NP. To show NP-hardness, we reduce from \textsc{x3c}, see the definition in Section~\ref{sec:prelim}.

Let $I$ be an instance of the \textsc{x3c}. We construct an instance $I'$ of \textsc{sr-strat} with $G=(V,E)$ as follows. For each set $C_j$, we add seven agents $c_j^1,c^2_j,c_j^3,d_j^1,d_j^2,d_j^3,x_j$ called the \emph{gadget for $C_j$}. For each element $a_i$, we add an agent also denoted by $a_i$. Finally, we add $2n$ selector agents $s_1,\dots s_{2n}$, see Figure~\ref{fig:st_strat}.

The preference lists are defined as follows:
\begin{itemize} \itemsep0em
    \item the preference lists of agents $a_i$ ($1\leq i\leq 3n$) and $s_j$ ($1\leq j\leq 2n$) are not fixed,
    \item if $C_j=(a_{j_1},a_{j_2},a_{j_3})$ is the $j$-th set in $I$, then:\\
\begin{minipage}{0.48\textwidth}
\vspace{-0.3cm}
\begin{align*}
    x_j\colon{}&{}  [S]>c_j^1>c_j^2>c_j^3>d_j^1>\dots\\
    c_j^1\colon{}&{} a_{j_1}>x_j>c_j^2>d_j^1>c_j^3>\dots\\
    c_j^2\colon{}&{} a_{j_2}>x_j>c_j^3>c_j^1>\dots \\
    c_j^3\colon{}&{} a_{j_3}>x_j>c_j^1>c_j^2>\dots
\end{align*}
\end{minipage}
\hfill
\begin{minipage}{0.48\textwidth}
\vspace{-0.3cm}
\begin{align*}
    d_j^1\colon{}&{} x_j>d_j^2>c_j^1>d_j^3>\dots \\
    d_j^2\colon{}&{} d_j^3>d_j^1>\dots\\
    d_j^3\colon{}&{} d_j^1>d_j^2>\dots
\end{align*}
\end{minipage}

where $[S]$ denotes $s_1>\dots>s_{2n}$ and `$>\dots$' means that the rest of the agents are ranked in an arbitrary order.
\end{itemize}

\begin{figure}[t!]
    \centering
    \includegraphics[width=0.4\textwidth]{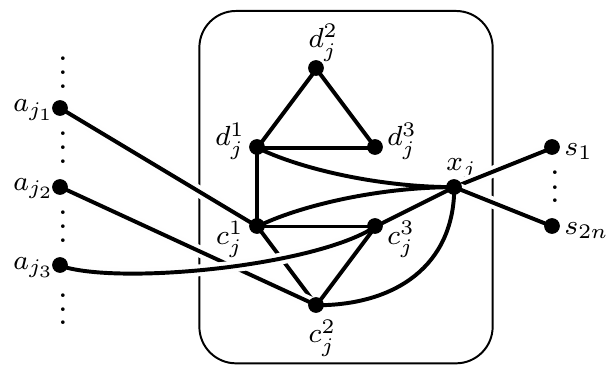}
    \caption{Gadget corresponding to set $C_j=\{a_{j_1},a_{j_2},a_{j_3}\}$ in Theorem~\ref{sr-strat} with the first few edges on the preference lists of the vertices.}
    \label{fig:st_strat}
\end{figure}

We claim that there is an exact 3-cover in $I$ if and only if the preferences of the remaining agents can be set in a way such that there is a stable matching $M$ in $G$. Suppose there is an exact 3-cover $\mathcal{C}'\subseteq\mathcal{C}$. Then, we make construct a matching $M$ as follows. We match the agents in $S=\{ s_1,\dots,s_{2n}\}$ to the agents $x_j$ that correspond to the 3-sets not in $\mathcal{C}'$. In the gadgets for these sets, we add the edges $c_j^2c_j^3, d_j^2d_j^3$ and $c_j^1d_j^1$ to the matching. In the gadgets of the sets in $\mathcal{C}'$, we add the edges $a_{j_1}c_j^1$, $a_{j_2}c_j^2$, $a_{j_3}c_j^3$, $xd_j^1$, and $d_j^2d_j^3$ to the matching. Finally, we set the preferences of the $a_i$s and $s_j$s such that each prefers its partner in $M$ the most. 

Suppose there is a blocking edge $uv$. It is clear that both $u$ and $v$ must be from the same set gadget, since each $a_i$ and $s_j$ is with their best partner, everyone is matched, and each agent in a set gadget prefers everyone in their gadget to anyone from another gadget. If $u$ and $v$ are in a gadget for a set $C_j \in \mathcal{C}'$, then each $c_j^i$ and $d_j^1$ is with their best partner, so they cannot block. Moreover, $x_j$ is with the best partner from the gadget that is not a $c_j^i$. Therefore, only $d_j^2d_j^3$ could block, but it is in $M$. If $u$ and $v$ are in a gadget for a set $C_j \notin \mathcal{C}'$, then $x_j$ is matched with its best partner, so it cannot block with anyone. It is straightforward to check that the edges $c_j^1c_j^2$, $c_j^3c_j^1$, $d_j^1d_j^2$, and $d_j^3d_j^1$ are all dominated, so no edge can block. 

Now suppose that there are preferences for the elements and selectors such that there exists a stable matching $M$ in $G$. We can assume that the elements and selector agents like their partner the most if there exists one. If a gadget for a set $C_j$ would get no selector and no element agent from outside, then $c_j^1$ and $x_j$ must be matched together. However, then there is no stable matching among the vertices $d_j^1$, $d_j^2$ and $d_j^3$. If the gadget for $C_j$ gets a selector, then we may assume that it is matched with $x_j$. In this case, there is a stable matching inside the gadget consisting of the edges $c_j^2c_j^3$, $c_j^1d_j^1$, and $d_j^2d_j^3$. If the gadget for $C_j$ gets no selector, then it must get all three of the corresponding element agents, as otherwise neither $x_j$ nor $c_j^1$ is matched to $d_j^1$, implying that there is no stable matching among the vertices $d_j^1$, $d_j^2$ and $,d_j^3$. Indeed, if only one or two element agents would get matched, then the preferences of $x_j$ imply that it must be matched to a remaining $c_j^i$ since the remaining $c_j^i$ prefers $x_j$ the most, and $x_j$ does not get a selector. If $x_j$ was matched to $c_j^2$ or $c_j^3$, then $x_jc_j^1$ would block, so it is matched to $c_j^1$.

In summary, for each set gadget, we must match either at least one selector agent or all three element agents to it in any stable matching. Therefore, if there is a stable matching $M$, then all set gadgets in $M$ must get either a selector or three element agents. But there are $2n$ selectors only, so there must be at least $n$ set gadgets that can only get element agents. By the above observation, each must get exactly three of those, so these $n$ sets must form an exact cover.
\end{proof}

The following corollary provides a different interpretation of Theorem~\ref{sr-strat}.

\begin{Cor}
\textsc{srt} is NP complete even if the graph is complete and each preference list is either strict or a single tie.
\end{Cor}

As a positive result, we show that the problem becomes tractable when the preferences are fixed for agents forming an independent set in the graph.

\begin{Th}
If the preferences are only fixed on an independent set, then \textsc{sr-strat} is polynomial-time solvable. 
\end{Th}
\begin{proof}
Let the agents with fixed preference lists be $v_1,\dots,v_k$. For $i\in[k]$, we do the following. Let $v_i$ choose its best available partner $u_i$ if there is any, add $u_iv_i$ to $M$, and set $u_i$'s preference such that $v_i$ is the best. Then find an inclusionwise maximal matching among the rest of the vertices, add it to $M$, and set the preference of each remaining vertex in a way that it prefers its partner the most. It follows from the construction that the matching thus obtained is stable. 
\end{proof}

\subsection{Fixed entries in the preference lists}
\label{sec:fixed}

We move on to problems where the preference lists of the agents have some fixed places that cannot be modified. The setting we concentrate on is the following. We are given a graph $G=(V, E)$ with vertex capacities $q(v)$ for $v\in V$ and an inclusionwise maximal feasible $q$-matching $M\subseteq E$. For a saturated vertex $v$ and edge $e\in M(v)$, we have a fixed number $p_v(e)$ denoting that $e$ is the $p_v(e)$-th worst edge in $v$'s ranking on the incident edges. Furthermore, a similar $p_v(e)$ value might be given for other vertex-edge pairs as well. Note that the $p_v(e)$ values for $e\in M(v)$ where $v$ is non-saturated are not interesting in the sense that no edge can be dominated at such a vertex. The task is to determine whether there exists an extension of these partial preference lists such that the given $q$-matching $M$ is stable. Such an extension is called a \emph{stable preference extension}. In a sense, this problem is the inverse of the stable $q$-matching problem, where the preferences are set in advance and the task is to find a stable $q$-matching $M$. 

\decprob{inc-max-spe}
{
A graph $G=(V,E)$, vertex capacities $q\in\mathbb{Z}_+^V$, an inclusionwise maximal feasible $q$-matching $M\subseteq E$, preference values $p_v(e)$ for each saturated vertex $v\in V$ and $e\in M(v)$, and preference values for some other vertex-edge pairs.
}
{
Does there exist a preference extension for which $M$ is stable?
}

We show that the problem of finding a stable preference extension can be reduced to computing a maximum stable matching in a bipartite graph. 

\begin{Th}\label{thm:ext1}
For any graph $G=(V,E)$, $q$-matching $M$ and partial preference $p$ which fixes $p_v(e)$ for each saturated vertex $v$ and each edge $e \in M$ containing $v$, \textsc{inc-max-spe} an be solved in polynomial time.  
\end{Th}
\begin{proof}
We reduce the problem to deciding if there exists a matching that covers a given side of a bipartite graph $H=(U,W;F)$ entirely. We may assume that there is no edge $e\in E \setminus M$ whose rank is fixed on both of its vertices in such a way that $e$ necessarily blocks $M$, as otherwise the answer is clearly `NO'. Also, we may assume that no edge $e\in E \setminus M$ has a fixed rank at a vertex where all edges in $M$ have better rank, because then $e$ is dominated and can be removed.

We construct a bipartite graph $H=(U,W;F)$ as follows. For each saturated vertex $v\in V$, we add vertices $v_1,\dots,v_{p_v(f)-1}$ to $U$ where $f$ is the worst edge for $v$ in $M$. If $p_v(f)=1$, then we do not include any copies of $v$. Also, if there is an index $1\leq j\leq p_v(f)-1$ such that $j$ is already assigned to an edge, then we delete that copy of $v$ from the graph. For each edge $e \notin M$ we add a vertex $v_{e}$ to $W$. Finally, we add an edge between a copy $v_j$ of a vertex $v\in V$ and a vertex $v_{e}$ corresponding to edge $e\in E(v)$ if and only if $p_v(e)$ is not yet fixed. 

We claim that the preferences can be extended in such a way that $M$ is a stable $q$-matching if and only if $H$ admits a matching covering the vertices in $W$. First suppose there is a matching $N$ in $H$ covering all vertices in $W$. We allocate the unassigned preferences of the saturated vertices that are worse than its worst edge in any order to the edges that are matched to its copies in $N$. The remaining preferences are chosen arbitrarily from the free positions. Then $M$ will become a stable $q$-matching of $G$ with respect to those preferences, since any edge contains a vertex that is saturated with strictly better edges. 

For the other direction, suppose that it is possible to extend the preferences in $G$ such that $M$ becomes stable. If we take the matching $M$ in $H$ that matches $v_j$ to the vertex corresponding to its $j$'th worst edge, then $M$ covers all vertices in $W$ as an uncovered vertex would correspond to a blocking edge.

As a matching covering the vertices in $W$ can be found in polynomial time, the theorem follows.
\end{proof}

As a slightly more general problem, suppose that lower and upper bounds $\ell_v(e)$ and $u_v(e)$ are given for each vertex $v\in V$ and edge incident to $v$. The goal is to extend the preferences in such a way that $M$ is stable and $\ell_v(e)\leq p_v(e)\leq u_v(e)$ holds for each vertex $v$ and each edge $e$ adjacent to $v$. 

\begin{Th}\label{thm:low}
If only lower bounds are given, then \textsc{inc-max-spe} can be solved in polynomial time.  
\end{Th}
\begin{proof}
We give an algorithm consisting of four  phases.
\smallskip

\noindent \textit{Phase 1:} 
We construct a bipartite graph $H=(U,W;F)$ similarly as in the proof of the unbounded case (Theorem~\ref{thm:ext1}). For each saturated vertex $u\in V$, we add $p_u(f)-k_u$ copies of $u$ to $U$, where $f$ is the worst edge for $u$ in $M$ and $k_u$ is the number of fixed values in the first $p_u(f)$ places. Also, for each edge $e\in E\setminus M$ that is not dominated already, we add a vertex $v_e$ to $W$. We add an edge between a copy $u_j$ of a vertex $u\in V$ and a vertex $v_{e}$ corresponding to an edge $e\in E(u)$ if and only if $p_u(e)$ is not yet fixed and $j\geq\ell_u(e)$. 

We check whether there exists a matching that covers all vertices of $W$; such a matching corresponds to the domination of the edges in $E\setminus M$ not yet dominated. If no such matching exists, then the algorithm stops with the answer `NO', the preferences cannot be extended in such a way that $M$ is a stable $q$-matching and the lower bounds are met. Otherwise, let $N_1$ be a matching covering $W$. 
\smallskip

\noindent \textit{Phase 2:}
We extend the bipartite graph $H$ by adding another vertex to $W$ for each edge that was not dominated originally, and two vertices for every edge that was not considered in Phase 1. That is, for each edge $e=xy\in E$, there are two corresponding vertices $v^x_e$ and $v^y_e$ in $W$. We modify the edges of $H$ in such a way that $v^x_e$ is only connected to the copies $x_i$ of $x$ and $v^y_e$ is only connected to the copies $y_j$ of $y$ where $i\geq\ell_x(e)$ an $j\geq\ell_y(e)$, respectively. Finally, we construct a matching $N_2$ in the extended bipartite graph by adding $v^u_eu_j$ to $N_2$ for each $v_eu_j\in N_1$.
\smallskip

\noindent \textit{Phase 3:}
Starting from the matching $N_2$, we check whether there exists a matching that covers all vertices of $U$ using the Hungarian method \cite{kuhn1955hungarian}. If no such matching exists, then the preferences cannot be extended in such a way that $M$ is a stable $q$-matching and the lower bounds are met, and the algorithm stops with the answer `NO'. If there exists such a matching, then a nice feature of the algorithm is that once a vertex becomes covered, it remains so throughout. Hence, for each edge $e=xy$, one of the vertices $v^x_e$ and $v^y_e$ is covered. We denote the resulting matching by $N_3$.
\smallskip

\noindent \textit{Phase 4:}
We further extend $H$ by adding all the copies of each vertex $u\in V$ to $U$ that have not been added in Phase 1. Also, we extend the set of edges such that there is an edge $u_jv^u_{e}$ for each $u\in V$, $e\in E(u)\setminus M(u)$ and $j\geq \ell_u(e)$. Consider the vertices corresponding to edges $e$ that were isolated until the addition of the new edges, that is, satisfied $\ell_u(e)\geq \min \{ p_u(f) \mid f\in M(u)\}$. We check if these vertices can all be matched to vertices not covered by $N_3$ by running the Hungarian algorithm. If no such matching exists, then there is no preference extension satisfying the lower bounds at all, and the algorithm outputs `NO'. 

Otherwise, let $N_4$ denote the matching obtained. We claim that the matching $N_3\cup N_4$ can be extended to a matching in such a way that it covers all vertices in $H$. To see this, observe that only lower bounds are given and for each remaining copy $v^u_{e}$ of an edge, the lower bound $\ell_u(e)$ is smaller than the indices of the uncovered copies of $u$. Indeed, all $v^u_e$ vertices with $l_u(e)\ge \min \{ p_u(f) \mid f\in M(u)\}$ are matched by $N_4$, and for each $u\in V$ the copies $u_j$ with $j\le \min \{ p_u(f) \mid f\in M(u)\}$ are matched by $N_3$. Therefore, the graph induced by the unmatched vertices consists of a union of complete bipartite graphs with equal sized color classes, so it admits a perfect matching $N_5$.
\smallskip

Setting the preferences according to $N_3\cup N_4\cup N_5$ gives a preference extension that satisfies all conditions.
\end{proof}

\begin{Rem}
Theorems~\ref{thm:ext1} and~\ref{thm:low} can be straightforwardly extended to hypergraphs. That is, given a hypergraph $H=(V, \mathcal{E})$, vertex capacities $q\in\mathbb{Z}_+^V$, a feasible $q$-matching $M\subseteq\mathcal{E}$, preference values $p_v(f)$ for $v\in V$ and $f\in M(v)$, preference values for some other vertex-hyperedge pairs, and lower bounds $\ell_v(f)$ for $v\in V$, $f\in E(v)$, one can decide in polynomial time if there exists a preference extension for which $M$ is stable.
\end{Rem}

A further extension would be to consider $q$-matchings $M$ that are not necessarily inclusionwise maximal. Let $M$ be an arbitrary $q$-matching and suppose that the places in the preference lists of the agents are fixed for the edges of $M$ for each saturated vertex and maybe on some other vertex-edge pairs, too. The goal is to set the remaining preferences in such a way that there exists a stable $q$-matching $M'$ containing $M$. 

\decprob{forced-spe}
{
A graph $G=(V,E)$, vertex capacities $q\in\mathbb{Z}_+^V$, a feasible $q$-matching $\mu \subseteq E$, preference values $p_v(e)$ for each saturated vertex $v\in V$ and $e\in \mu (v)$, and preference values for some other vertex-edge pairs.
}
{
Does there exist a preference extension together with a stable $q$-matching $M$ for which $\mu \subseteq M$?
}

Surprisingly, this problem becomes NP-complete even for the stable marriage case.

\begin{Th}\label{thm:forced-spe}
 \textsc{forced-spe} is NP complete even for bipartite graphs with capacity one on each vertex.
\end{Th}
\begin{proof}
The problem is clearly in NP. To show NP-hardness, we reduce from \textsc{x3c}, see the definition in Section~\ref{sec:prelim}.

Let $I$ be an instance of \textsc{x3c}. We construct an instance $I'$ with $H=(U,W;F)$ of \textsc{forced-spe} as follows. Let $P=\{ p_1,\dots,p_{3n}\}$ be a set of agents corresponding to the ground set $S$, and $Q=\{ q_1,\dots ,q_m\}$ be a set of agents corresponding to the 3-sets. These two sets of agents form one color class of the bipartite graph which will be referred to as men. Agents in $X=\{ x_1,\dots ,x_m\}$, $Y=\{ y_1,\dots ,y_m\}$ and $Z=\{ z_1,\dots ,z_{3n}\}$ form the other color class of the bipartite graph which will be referred to as women. Finally, we add a special agent $u$ to the set of men and a special agent $w$ to the set of women. So $U=P \cup Q \cup \{ u\}$, $W=X\cup Y\cup Z \cup \{ w\}$.

The forced matching is $\mu=\{uw, p_iz_i \mid  i\in[3n]\}$. The fixed preferences of the agents on the edges of $\mu$ are the following: 
\begin{itemize} \itemsep0em
    \item for $u$, the rank of $M(u)$ is $m-n+1$,
    \item for a man $p_i$, the rank of $M(p_i)$ is $2$,
    \item for a woman $z_i$, the rank of $M(z_i)$ is set arbitrarily.
\end{itemize}
We further add the edges $\{ ux_i, x_iq_i, q_iy_i \mid i\in[m]\}$ to $F$, and for each 3-set $C_j=\{ s_{i_1},s_{i_2},s_{i_3}\}$ of $I $ we add the edges $y_jp_{i_1}$, $y_jp_{i_2}$ and $y_jp_{i_3}$. 

We claim that there is an exact 3-cover in $I$ if and only if the preferences can be extended in a way that there is a stable matching $M$ containing $\mu$. Let us first suppose that $C_{i_1},\dots ,C_{i_n}$ is a solution of the \textsc{x3c} instance $I$. Let $M$ be the matching obtained by adding the edges $\{x_{i_j}q_{i_j}\mid j\in[n]\}$ and the edges $\{q_ky_k\mid k\notin \{ i_1,\dots ,i_n\}\}$ to $\mu$. The preferences are set such that the $m-n$ places worse than that of $uw$  for $u$ are given to the edges $\{ux_k\mid k\notin \{ i_1,\dots ,i_n\}\}$. For $k\notin \{ i_1,\dots ,i_n\}$, we set $q_k$ to prefer $y_k$ to $x_k$.  For $j\in[n]$, we set $q_{i_j}$ to prefer $x_{i_j}$ to $y_{i_j}$ and $x_{i_j}$ to prefer $q_{i_j}$ to $u$. Also for each $i=1,\dots 3n$, we set the worst position of $p_i$ to be the $y_j$ corresponding to the set $C_j\in \{ C_{i_1},\dots C_{i_n} \}$ containing the element $a_i$.
Since $C_{i_1},\dots ,C_{i_n}$ is a partition, there is no collision when doing this. Then, all edges not in $M$ are dominated: each edge of type $ux_j$ is dominated at either $u$ or $x_j$, each edge $q_jx_j\notin M$ is dominated at $q_j$, each edge $q_jy_j\notin M$ is dominated at $q_j$, and each edge $p_iy_j$ is dominated at either $y_j$ or $p_i$.

For the other direction, suppose there is a preference extension and a matching $M$ with $\mu \subseteq M$ that satisfy the conditions. For any triple $x_i$, $q_i$, and $y_i$, one of these vertices must be uncovered in $M$ since the only possible partner of both $x_i$ and $y_i$ is $q_i$ as the other neighbors are already matched in $\mu$. There are $m-n$ women $x_i$ such that $x_i$ is worse for $u$ than $w$. At the same time, there must be $n$ women $x_{i_1},\dots ,x_{i_n}$ that are better, so for the edges $ux_{i_j}$ to be dominated we must have $x_{i_j}q_{i_j}\in M'$. Then $y_{i_j}$ cannot be matched for $j=1,\dots n$, hence the edges that join these vertices to $P$ must be dominated at the vertices of $P$. But there are at least $n$ such $y_{i_j}$-s with altogether $3n$ edges incident to them and each $p_i\in P$ can dominate at most one such edge, therefore $C_{i_1},\dots ,C_{i_n}$ is necessarily a partition.
\end{proof} 

\begin{Rem}
If the size of the maximal matching $M$ containing $\mu$ is larger than the size of $\mu$ only by a constant, then the problem can be solved in polynomial time.
\end{Rem}

\begin{Rem}
The proof of Theorem~\ref{thm:forced-spe} implies that the bipartite variant of \textsc{forced-spe} when the preferences are set on $M$ only for one class of the agents is also NP-complete.
\end{Rem}

\begin{Rem}
If the graph is a complete bipartite graph, or a complete graph and the preferences are only set on the edges of $\mu$, then the problem can be solved in polynomial time. First, check whether the edges that connect vertices matched by $\mu$ can all be dominated as we did before. If this is possible, then find a perfect matching $M$ containing $\mu$ and set the edges of $M\setminus \mu$ as the best edges for every agent. Then $M$ is clearly a stable matching.  
\end{Rem}

Now suppose that a set $\tau$ of forbidden edges is given instead of forced ones while some preferences of the agents are fixed. The question is whether the preference lists can be extended such that there exists a stable matching $M$ disjoint from $\tau$.

\decprob{forbidden-spe}
{
A graph $G=(V,E)$, vertex capacities $q\in\mathbb{Z}_+^V$, a set $\tau \subseteq E$, and preference values for some other vertex-edge pairs.
}
{
Does there exist a preference extension together with a stable $q$-matching $M$ for which $\tau\cap M=\emptyset$?
}

Interestingly, \textsc{forbidden-spe} turns out to be NP-hard even if there are no predefined preferences, each capacity is one and the graph is bipartite. In that case the problem is equivalent to the following: Given a bipartite graph $G=(A,B;E)$ where $E$ is partitioned into red edges $E_r$ and blue edges $E_b$, decide if there is a matching $M$ in the subgraph $G_r=(A,B;E_r)$ such that the vertices of $M$ cover every blue edge. We call this problem \textsc{red-blue-edge-cover}.
\decprob{red-blue-edge-cover}
{
A bipartite graph $G=(A,B;E)$, with $E=E_b\cup E_r$, $E_b\cap E_r=\emptyset$.
}
{
Is there a matching $M\subseteq E_r$, such that for each $e=uv\in E_b:$ $M(u)\ne \emptyset$ or $M(v)\ne \emptyset$?
}

\begin{Th}
\textsc{red-blue-edge-cover} is NP-complete.
\end{Th}
\begin{proof}
The problem is clearly in NP. To show NP-hardness, we reduce from \textsc{3sat}.

Let $\Phi$ be a \textsc{3sat} instance with clauses $C_1,\dots ,C_m$ and variables $x_1,\dots,x_n$. We construct a bipartite graph $G=(A,B;E)$ as follows. For each clause $C_j$, we add two vertices $c_j'$ and $c_j''$ to $A$. For each literal $x_i$ or $\overline{x}_i\in C_j$, we add a vertex $y_j^i$ or $z_j^i$, respectively, to $B$. We connect these five vertices with red edges, that is, a red $K_{2,3}$ subgraph is added for every clause. For each variable $x_i$, we add two vertices $y_i$ and $z_i$ to $A$, a selector vertex $s_i$ to $B$, and two red edges $y_is_i$ and $z_is_i$. Finally, for $i\in[n]$ we add blue edges between $y_i$ and $y_j^i$ for the values of $j$ for which $x_i\in C_j$, and between $z_i$ and $z_j^i$ for the values of $j$ for which $\overline{x}_i\in C_j$. 

We claim that $\Phi$ has a satisfying assignment if and only if $G$ admits a matching $M$ consisting of red edges whose end vertices cover every blue edge. Suppose there exists a satisfying assignment for $\Phi$. Then for each $i$, we add the red edge $y_is_i$ to $M$ if $x_i$ was true, otherwise the edge $z_is_i$. For each clause $C_j$, we pick a literal that is true and add two red edges from $a_i'$ and $a_i''$ to the vertices corresponding to the other two literals. This way each blue edge will be covered. Indeed, if a blue edge is not covered at its vertex in $A$, then it corresponds to a false literal. However, in such a case the edge is covered in one of the $K_{2,3}$s.

For the other direction, let $M$ be a red matching that covers every blue edge. In every red $K_{2,3}$, there is a blue edge from one of its vertices in $B$ that is not covered by edges of the $K_{2,3}$, hence it has to be covered by an edge $y_is_i$ or $z_is_i$, and only one of these two edges can be in $M$. If we set the values of the variables such that $x_i$ is true if and only if $y_is_i\in M$, then we get a satisfying assignment. 
\end{proof}

\medskip
\paragraph{Acknowledgement.} The work was supported by the Lend\"ulet Programme of the Hungarian Academy of Sciences -- grant number LP2021-1/2021 and by the Hungarian National Research, Development and Innovation Office -- NKFIH, grant numbers FK128673 and TKP2020-NKA-06.

\bibliographystyle{abbrv}
\bibliography{preference_extensions}
 
 \end{document}